%% file: main.tex
\documentclass[11pt]{article}
\usepackage[margin=1in]{geometry}

\pdfoutput=1

\usepackage{balance}
\usepackage[utf8]{inputenc} % allow utf-8 input
\usepackage[T1]{fontenc}    % use 8-bit T1 fonts
\usepackage{lmodern}
\usepackage{url}            % simple URL typesetting
\usepackage{booktabs}       % professional-quality tables
\usepackage{amsfonts}       % blackboard math symbols
\usepackage{nicefrac}       % compact symbols for 1/2, etc.
\usepackage{microtype}      % microtypography
\usepackage{amsmath}
\usepackage{inputs/color-edits}
\usepackage{subcaption}

\usepackage{bm}
\usepackage{booktabs}
\usepackage{amssymb}
\usepackage{algorithm}
\usepackage[noend]{algpseudocode}
\usepackage{algorithmicx}
\usepackage{mathtools}
\usepackage{comment} % enables comment environment
\usepackage[most]{tcolorbox}
\usepackage{thm-restate}
\usepackage{xfrac}
\usepackage{multirow}
\usepackage[labelfont=bf, font=small,aboveskip=0pt,belowskip=0pt]{caption}
\usepackage{bm}
\usepackage{newfloat}
\usepackage{enumitem}
\usepackage{xpatch}
\usepackage{flushend}
\usepackage{amsthm}

\usepackage[labelfont=bf, font=small,aboveskip=0pt,belowskip=0pt]{caption}

\allowdisplaybreaks

\definecolor[named]{ACMBlue}{cmyk}{1,0.1,0,0.1}
\definecolor[named]{ACMYellow}{cmyk}{0,0.16,1,0}
\definecolor[named]{ACMOrange}{cmyk}{0,0.42,1,0.01}
\definecolor[named]{ACMRed}{cmyk}{0,0.90,0.86,0}
\definecolor[named]{ACMLightBlue}{cmyk}{0.49,0.01,0,0}
\definecolor[named]{ACMGreen}{cmyk}{0.20,0,1,0.19}
\definecolor[named]{ACMPurple}{cmyk}{0.55,1,0,0.15}
\definecolor[named]{ACMDarkBlue}{cmyk}{1,0.58,0,0.21}
\usepackage[bookmarksnumbered,unicode]{hyperref}
\hypersetup{colorlinks,
  linkcolor=ACMRed,
  citecolor=ACMPurple,
  urlcolor=ACMDarkBlue,
  filecolor=ACMDarkBlue}

\usepackage{zi4}
\usepackage[frozencache,cachedir=.]{minted}

\input{inputs/utils}

\input{inputs/macros}

\begin{document}\sloppy

\title{Parallel Graph Algorithms in Constant Adaptive Rounds: \\ Theory meets Practice
\thanks{
This is the full version of a paper in PVLDB (to be presented at VLDB'21).
The authors can be contacted at soheil@cs.umd.edu, laxman@mit.edu, and
\{esfandiari, jlacki, mirrokni, wschudy\}@google.com}
}
\author{
  Soheil Behnezhad\\University of Maryland
  \and
  Laxman Dhulipala \thanks{Work done as a student researcher at Google Research.} \\MIT CSAIL
  \and
  Hossein Esfandiari\\Google Research
  \and
  Jakub Łącki\\Google Research
  \and
  Vahab Mirrokni\\Google Research
  \and
  Warren Schudy\\Google Research
}

\maketitle
\input{inputs/abstract}
\input{inputs/intro}

\input{inputs/model}

\input{inputs/mst}

\input{inputs/matching}

\input{inputs/experiments_new}

\input{inputs/relwork}

\input{inputs/conclusion}

\section*{Acknowledgements}
We would like to thank Greg Cipriano, Milo Martin, and Florentina
Popovici for their work on the key-value store used in this paper, and
Raimondas Kiveris for his contributions to earlier projects which
inspired this work. We would also like to thank the anonymous
reviewers for their helpful feedback and suggestions.

\bibliographystyle{abbrv}
\bibliography{ref}

\appendix
\input{inputs/proofs}

\input{inputs/flight}

\end{document}

%% file: inputs/utils.tex
\newtheorem{theorem}{Theorem}
\newtheorem{lemma}{Lemma}[section]
\newtheorem{proposition}[lemma]{Proposition}
\newtheorem{corollary}[lemma]{Corollary}

\newtheorem{definition}[lemma]{Definition}

\addauthor{sb}{blue} % sb for Soheil
\addauthor{ld}{red} %ld for Laxman
\addauthor{kl}{purple} %kl for Kuba
\addauthor{he}{gray} %he for Hossein
\addauthor{ws}{orange} %ws for Warren
\addauthor{vm}{green} %vm for Vahab

\DeclareMathOperator{\poly}{poly}

\definecolor{mygreen}{RGB}{20,140,80}
\definecolor{linkcolor}{RGB}{0,0,230}
\definecolor{mylightgray}{RGB}{230,230,230}
\definecolor{verylightgray}{RGB}{245,245,245}

\algnewcommand{\IIf}[1]{\State\algorithmicif\ #1\ \algorithmicthen}
\algnewcommand{\EndIIf}{\unskip\ \algorithmicend\ \algorithmicif}

\newcommand{\smparagraph}[1]{
\vspace{0.2cm}
\noindent \textbf{#1}
}

\newcounter{myalgctr}
\newtcolorbox{OuterBox}[1][]{%
    breakable,
    enhanced,
    frame hidden,
    interior hidden,
    left=-5pt,
    right=-5pt,
    top=-5pt,
    float=p,
    boxsep=0pt,
    arc=0pt
#1}%

\newtcolorbox{InnerBox}[1][]{%
    enforce breakable,
    enhanced,
    colback=gray,
    colframe=white,
#1}%

\newenvironment{tbox}{
\vspace{0.2cm}
\begin{tcolorbox}[
                  enhanced,
                  boxsep=2pt,
                  left=1pt,
                  right=1pt,
                  top=4pt,
                  boxrule=1pt,
                  arc=0pt,
                  colback=white,
                  colframe=black,
	              breakable,
	              floatplacement=t,
	              float
                  ]%%
}{
\end{tcolorbox}
}

\newcommand{\tboxhrule}[0]{\vspace{0.1cm} {\color{black} \hrule} \vspace{0.2cm}}

\newenvironment{titledtbox}[1]{\begin{tbox}#1 \tboxhrule}{\end{tbox}}

\newenvironment{tboxalg}[1]{\refstepcounter{myalgctr}\begin{titledtbox}{\textbf{Algorithm \themyalgctr.} #1}}{\end{titledtbox}}

%% file: inputs/macros.tex
\newcommand{\MPC}[0]{\ensuremath{\mathsf{MPC}}}
\newcommand{\AMPC}[0]{\ensuremath{\mathsf{AMPC}}}
\newcommand{\PRAM}[0]{\ensuremath{\mathsf{PRAM}}}
\newcommand{\EREW}[0]{\ensuremath{\mathsf{EREW}}}
\newcommand{\CREW}[0]{\ensuremath{\mathsf{CREW}}}

\newcommand{\MultiPRAM}[0]{MultiPrefix \ensuremath{\mathsf{PRAM}}}
\newcommand{\defn}[1]{\emph{\textbf{#1}}}

\newcommand{\myparagraph}[1]{\smallskip\noindent {\bf #1.}}
\newcommand{\etal}[0]{\textit{et al.}}

\newcommand{\threep}[0]{ternary treap}

\newcommand{\boruvka}{Bor\r{u}vka}
\definecolor{revise-color}{rgb}{0.76, 0.23, 0.13}
\definecolor{mygreen}{rgb}{0.0, 0.5, 0.0}
\definecolor{amaranth}{rgb}{0.9, 0.17, 0.31}
\newcommand{\revised}[1]{{#1}}
\newcommand{\shep}[1]{{#1}}

%% file: inputs/abstract.tex
\begin{abstract}
We study fundamental graph problems such as {\em graph connectivity}, {\em minimum spanning forest} (MSF), and approximate {\em maximum (weight) matching} in a distributed setting. In particular, we focus on the {\em Adaptive Massively Parallel Computation} (\AMPC{}) model, which is a theoretical model that captures MapReduce-like computation augmented with a distributed hash table.

We show the first \AMPC{} algorithms for all of the studied problems that run in a constant number of rounds and use only $O(n^\epsilon)$ space per machine, where $0 < \epsilon < 1$.
Our results improve both upon the previous results in the \AMPC{}
model, as well as the best-known results in the \MPC{} model, which is
the theoretical model underpinning many popular distributed computation
frameworks, such as MapReduce, Hadoop, Beam, Pregel and Giraph.

\shep{Finally, we provide an empirical comparison of the algorithms in the
\MPC{} and \AMPC{} models in a fault-tolerant distriubted computation environment.} We empirically evaluate our algorithms on a
set of large real-world graphs and show that our \AMPC{} algorithms
can achieve improvements in both running time and round-complexity
over optimized \MPC{} baselines.
\end{abstract}

%% file: inputs/intro.tex
\section{Introduction}

The \MPC{} model has been extensively studied in theory in recent years~\cite{  DBLP:conf/spaa/AhnG15, DBLP:conf/stoc/AndoniNOY14, andoniparallel, DBLP:conf/soda/AssadiBBMS19, DBLP:conf/nips/BateniBDHKLM17, DBLP:conf/spaa/BehnezhadDETY17, DBLP:journals/corr/abs-1802-10297, DBLP:conf/focs/BehnezhadHH19, DBLP:conf/stoc/CzumajLMMOS18, DBLP:conf/podc/GhaffariGKMR18,  koutris2018algorithmic, DBLP:conf/spaa/LattanziMSV11, DBLP:conf/spaa/RoughgardenVW16}, and its theoretical capabilities and limitations are relatively well-understood.
In the context of graph algorithms, a significant limitation of the model is, roughly speaking, the fact that initially each node only knows its immediate neighbors, and exploring a larger neighborhood requires multiple rounds.
This phenomenon is formalized in the widely believed \textsc{1-vs-2-Cycle} conjecture~\cite{DBLP:conf/focs/GhaffariKU19}, which states that in the \MPC{} model, distinguishing between a cycle of length $n$ and two cycles of length $n/2$ requires $\Omega(\log n)$ computation rounds.
The conjecture implies a number of conditional lower bounds for fundamental graph problems such as matching, independent sets, coloring, connectivity, etc., see \cite{DBLP:conf/focs/BehnezhadDELM19, DBLP:conf/focs/GhaffariKU19} or a recent result of \cite{DBLP:journals/corr/abs-2001-02191} for an extensive overview of the implications of the conjecture. The main motivation behind the \AMPC{} model is to alleviate such hardness results.

\begin{table*}
  \centering
\revised{
\scalebox{0.8}{
\begin{tabular}{l|c|c|c|c}
    \toprule
	\multirow{2}{*}{\textbf{Problem}} & \multicolumn{2}{c|}{\textbf{\AMPC{}} -- \textbf{This Paper}} & \textbf{\AMPC{}} -- \cite{AMPC} & \multirow{2}{*}{\textbf{\MPC{}}} \\
	\cline{2-3}
	 & Bound  & Implementation?  & Previous Bound& \\
    \midrule
	Connectivity & $O(1)$ & & $O(\log \log_{m/n} n)$ & $O(\log D + \log \log_{m/n} n)$~\cite{DBLP:conf/focs/BehnezhadDELM19}\\
	MST & $O(1)$ & \checkmark & $O(\log \log_{m/n} n)$ & $O(\log n)$\\
	Matching & $O(\log \log n)$ &  & --- & $\widetilde{O}(\sqrt{\log n})$~\cite{DBLP:conf/soda/GhaffariU19} \\
	Matching (with $O(m+n^{1+\epsilon})$ space*) & $O(1)$ & \checkmark & --- & $\widetilde{O}(\sqrt{\log n})$~\cite{DBLP:conf/soda/GhaffariU19}\\
	MIS &  & \checkmark & $O(1)$ & $\widetilde{O}(\sqrt{\log n})$~\cite{DBLP:conf/soda/GhaffariU19} \\
	\textsc{1-vs-2-Cycle}   &  & \checkmark & $O(1)$ & $O(\log n)$ (folklore) \\
    \bottomrule
  \end{tabular}
}
	}
	\caption{\label{tab:results}\revised{ Summary of our results, compared to the results of~\protect\cite{AMPC}.
	We assume the available space per machine is $n^\delta$ for some
  constant $\delta < 1$. We consider the setting where the space per
  machine is sublinear in the number of vertices of the graph, that is
  $S = O(n^{\epsilon})$. All results require $\tilde{O}(m)$ total
  space over all machines (except the fourth row). $D$ denotes the diameter of the graph. We note that an algorithm for the minimum spanning tree problem also solves the connectivity problem.}}
\end{table*}

At a high level, in the \MPC{} model computation is distributed among a number of machines and proceeds in synchronous rounds.
The space available on each machine is assumed to be much smaller than the size of the input.
Within each round, each machine may send (a limited number of) messages to other machines, and all these messages constitute the input to the following round, i.e., they are delivered only when the next round begins.
Hence, each machine can only access a limited part of the input at each round, that is only the messages that it receives.

The \AMPC{} model \emph{extends} the \MPC{} model by storing all messages sent in a single round in a distributed read-only hash table (also known as a distributed \emph{key-value store}).
In the following round, all machines have \emph{random read-access} to the hash table, subject to the same constraints on the amount of communication as in the \MPC{} model.
As shown in~\cite{AMPC}, this extension is useful in designing graph algorithms, as, intuitively, it allows a machine knowing the identity of a vertex to explore its local neighborhood.
(On the other hand, Charikar~\etal{}~\cite{lowerboundAMPC} showed that some of the known unconditional lower bounds in the \MPC{} model apply to the \AMPC{} model as well.)

The \AMPC{} model is inspired by empirical studies of large-scale graph computations that utilize a distributed hash table~\cite{nips17,cc-beyond}.
In the paper introducing the model, Behnezhad~\etal{}~\cite{AMPC}
argued that the \AMPC{} model is realistic, that is it can be
implemented by augmenting a parallel processing framework, such as
MapReduce, Hadoop or Beam, with a distributed hash table.
At the same time, the hybrid system would maintain the fault-tolerance
properties of the framework used and should not suffer from query contention.

Implementing the \AMPC{} model efficiently relies on using an
efficient distributed hash table. Developing a high performance
distributed hash table has become a popular topic of research in
the systems community in recent years.
The existing implementations are able to support hundreds of millions of queries per second,
with query latency as low as a few
microseconds~\cite{dragojevic2017rdma}. These technologies are
increasingly available today, both at a hardware level~\cite{
infiniband, omnipath}, and at the level of software systems such as
eRPC~\cite{DBLP:conf/nsdi/KaliaKA19}, and key-value databases built
over RDMA~\cite{dragojevic2014farm, mitchell2013rdmakvstores}.

The results of Behnezhad~\etal{}~\cite{AMPC} show that that using \AMPC{} model
allows one to give algorithms running in significantly fewer rounds compared to
the best solutions in the \MPC{} model. However, up until today it has
not been studied how these gains would translate to the empirical
running times.  In this paper we provide the first empirical
evaluation of the \AMPC{} model, and further improve the theoretical
results by Behnezhad~\cite{AMPC}, by giving \AMPC{} algorithms for
basic graph problems that run in only $O(1)$ rounds.

\subsection{Our Contribution}

In this paper, we study several fundamental graph problems in the
\AMPC{} model, present the first distributed algorithms for these
problems that run in a constant number of rounds, and provide the
first empirical evaluation of this model. Our
results along with a comparison with the state-of-the-art both in the
\MPC{} model and the \AMPC{} model are provided in
Table~\ref{tab:results}. Below, we elaborate more on these results.

\smparagraph{Connected components.} Our first result is an \AMPC{}
algorithm for computing the {\em connected components} of an
undirected graph in $O(1)$ rounds.  We note that finding
connected components, apart from being a basic graph problem, has a
number of practical
applications and is one of the most commonly used graph-processing
analytics~\cite{DBLP:journals/pvldb/SahuMSLO17}.

Note that before the \AMPC{} model has been introduced, the problem of
finding connected components has been studied in a setting very
similar to the \AMPC{} model itself.
Namely, \cite{cc-beyond} gave an algorithm running in $O(\log n \log
\log n)$ rounds and showed its good practical performance.
However, it has been later shown that the empirical running time can be further improved by using an \MPC{} algorithm~\cite{cc-contractions}.
We note that in the \MPC{} model, the best known algorithm for finding
connected components runs in $O(\log D + \log \log_{m/n} n)$ rounds,
where $D$ is the diameter of the input
graph~\cite{DBLP:conf/focs/BehnezhadDELM19}.

\smparagraph{Minimum spanning forest.} Our next notable result is an
\AMPC{} algorithm for computing  {\em minimum spanning forest} (MSF)
in a constant number of rounds.
This problem has a number of applications in clustering~\cite{nips17}.
In particular, one can use this algorithm together with a simple sorting step, and our connectivity algorithm to find any desired level of a
single-linkage hierarchical clustering~\cite{DBLP:conf/uai/ZadehB09}.
Similarly to the problem of connectivity, finding MSF can be solved in $O(\log n)$ rounds in the $\MPC{}$ model (or slightly faster, if one assumes that the diameter of each tree in the minimum spanning forest is limited~\cite{andoniparallel,DBLP:conf/focs/BehnezhadDELM19}).

Our algorithms for finding connected components and MSF that run in $O(1)$ rounds
directly improve upon the state-of-the-art \AMPC{} algorithms by Behnezhad~\etal{}~\cite{AMPC} that run in $O(\log \log_{m/n} n)$ rounds.
At the same time, under \textsc{1-vs-2-Cycle} conjecture~\cite{DBLP:conf/focs/GhaffariKU19},
both problems require $\Omega(\log n)$ rounds in the \MPC{} model (or $\Omega(\log D)$ rounds if we parametrize the problem by graph diameter $D$~\cite{DBLP:conf/focs/BehnezhadDELM19}).
Hence, with respect to the number of rounds, our results improve significantly over the best \MPC{} solutions.

\smparagraph{Matchings.} Furthermore, we give two \AMPC{} algorithms for the maximal matching problem. Both algorithms require $O(n^\epsilon)$ space per machine where $\epsilon \in (0, 1)$ can be any arbitrarily small constant. The first algorithm uses $\widetilde{O}(m+n)$ total space and $O(\log\log n)$ rounds. The second algorithm uses $O(m+n^{1+\epsilon})$ total space and takes $O(1)$ rounds. The best known previous algorithm in this regime of space can be obtained by simulating the \MPC{} algorithm of \cite{DBLP:conf/soda/GhaffariU19} which requires $\widetilde{O}(\sqrt{\log n})$ rounds. We note that an $O(\log \log n)$ round \MPC{} maximal matching algorithm is also known \cite{DBLP:conf/focs/BehnezhadHH19}, which unlike ours, requires $\Omega(n)$ space per machine.

\smparagraph{Empirical evaluation.}
\shep{Finally, we provide the first empirical evaluation of the \AMPC{}
model in a fault-tolerant distributed setting.} We experimentally evaluate four
\AMPC{} algorithms for the Maximal Independent Set (MIS), Maximal
Matching (MM), \revised{Minimum Spanning Forest}, and 1-vs-2-Cycle problems, and
compare these algorithms with state-of-the-art \MPC{} baselines for
these problems.  We implemented all of our algorithms (both in \MPC{}
and \AMPC{} models) in Flume-C++~\cite{flumecpp}, which is a highly
optimized dataflow framework similar to Apache Beam.  The \AMPC{}
algorithms were additionally allowed to use a distributed hash table
service.
\revised{
We provide a detailed
experimental study of our implementations in a shared production
data center using a maximum of 100 machines.
We investigate the effect of enabling caching and multithreading
optimizations which enable good practical performance for \AMPC{}
algorithms. Finally, we study the performance of our \AMPC{}
algorithms relative to strong \MPC{} baselines and show significant
speedups for all problems that we study.
Our results provide promising evidence for the practical applicability
of algorithms in the \AMPC{} model.
}

%% file: inputs/model.tex
\section{Model}\label{sec:model}
\newcommand{\dht}{\mathcal{D}}
In this section we give a concise formal definition of the
\emph{Adaptive Massively Parallel Computation (\AMPC{})} model (the
model used in this paper is identical to the model
defined in~\cite{AMPC}).
In the \AMPC{} model, we are given an input of size $N$, which is
processed by a collection of $P$ machines each having space $S$. The
total space used by the computation is just the total space of all
machines, i.e., $T = S \cdot P$. We assume that $S =
\Theta(N^{\epsilon})$, where $\epsilon \in (0, 1)$ is a constant.
We note that while the running times of our algorithms often
depend on $1/\epsilon$ we often omit this dependency, since
$O(1/\epsilon) = O(1)$.\footnote{In practice $\epsilon$ is at least
$1/2$, i.e. the number of bytes of space on each machine is larger
than the number of machines.}

In addition to the input and machines, in the \AMPC{} model there is
collection of \emph{distributed hash tables (DHT)} that we denote by
$\dht_0, \dht_1, \dht_2, \ldots$.
Each hash table stores a collection of key-value pairs:
given a key a DHT returns all corresponding values. We require each
key-value pair to use a constant number of words. At the start of the
computation, the input data is stored in $\dht_0$ and uses a set of
keys known to all machines (e.g., consecutive integers).

An \AMPC{} computation consists of a number of \emph{rounds}. In the
i-th round, each machine can read data from $\dht_{i-1}$ and write to
$\dht_i$. Within a single round, each machine can make up to $O(S)$
reads (\emph{queries}) and emit $O(S)$ \emph{writes}, and perform
arbitrary computation. Each query and write queries/writes a single
key-value pair. The total \emph{communication} that a machine performs
per round is equal to the total number of queries and writes. A
salient aspect of the model is that the total communication of each machine within a round is bounded
by $O(S)$ which is strictly sublinear in the input size. We note that
all of the algorithms described in both~\cite{AMPC} and this paper
perform near-linear computation per machine in each round, making the
algorithms proposed in both papers work-efficient up to
poly-logarithmic factors. However, the \AMPC{} model does not place a
bound on the amount of of computation permitted on a machine, as this
is a usual assumption of the \MPC{} model.

\shep{
\smparagraph{Fault-Tolerance}
An important characteristic of the \AMPC{} model is that it is
amenable to \emph{fault tolerant} implementation, which is one of the key
features of the \MPC{} model, and one of the reasons underpinning its widespread
adoption. A fault tolerant implementation of \AMPC{} can be derived by
observing that each DHT can be made fault-tolerant~\cite{AMPC}.
}

\shep{
\smparagraph{Caching and Query Contention}
One natural concern related to implementing the \AMPC{} algorithms is the
possibility that all $S$ machines query for the same key, potentially causing
contention on a single machine, or set of machines storing the key.
The authors of~\cite{AMPC} provide an algorithmic justification for
how any set of queries made by the machines in a given round can be
handled with negligible overall contention by caching the results of
queries on each machine. In this paper, we empirically verify that
caching is required for good performance in the \AMPC{} model (see
Section~\ref{sec:miscasestudy}).
}

\smparagraph{Relationship to other models}
The \AMPC{} model is closely related to the \MPC{} model. Indeed,
every \MPC{} algorithm can be simulated by a corresponding \AMPC{}
algorithm in the same total space and round-complexity (\cite{AMPC}
provides a description of the simulation). Due to known simulations of
\PRAM{} algorithms on \MPC{}~\cite{DBLP:conf/isaac/GoodrichSZ11,DBLP:conf/soda/KarloffSV10}, the \AMPC{} model is also able to
simulate existing \PRAM{} algorithms in the \EREW{}, \CREW{}, \CREW{},
and \MultiPRAM{} models~\cite{Blelloch:2010:PA:1882723.1882748}.

%% file: inputs/mst.tex
\section{Connectivity \& Minimum Spanning Forest}\label{sec:msf}

In this section we give our algorithms for finding minimum spanning forest (MSF),
i.e., the minimum spanning tree of each connected component, as well as
for finding connected components. Some proofs from this section can be
found in Appendix~\ref{apx:msf}. Formally, we prove the following
result.

\begin{theorem}\label{thm:mstlowspace}
There is an \AMPC{} algorithm for computing minimum spanning forest
and connected components of an undirected weighted graph in $O(1)$ rounds
using total space $T = O(m + n \log^2 n)$, w.h.p.
\end{theorem}

Our result for connectivity can be actually obtained from the algorithm
for minimum spanning forest.
Indeed, once we find any spanning forest, the connected components
can be found by applying the forest connectivity algorithm of~\cite{AMPC}
which takes $O(1)$ rounds and uses $O(n \log n)$ queries.
Hence, in the following part of this section we focus on an algorithm
for computing MSF in $O(1)$ rounds of $\AMPC{}$. Our algorithm improves on the
$O(\log\log_{T/n} n)$ round connectivity and MSF algorithms of
Behnezhad \etal{}~\cite{AMPC} whenever $T = n^{1+o(1)}$.

Algorithm~\ref{alg:msf} provides the pseudocode for our MSF algorithm.
If the input graph is dense, i.e. $m = \Omega(n^{1+\epsilon})$, the
algorithm runs the algorithm of Behnezhad \etal{} which finishes in
$O((1/\epsilon) \log (1/\epsilon))$ rounds~\cite{AMPC}. To handle the case when $m =
o(n^{1+\epsilon})$, the algorithm ternarizes the graph, that is, replaces
every vertex of degree $k > 3$ with a cycle of length $k$, with each
edge associated with one vertex on the cycle. The ternarization step
ensures that every vertex has degree $\leq 3$, but will make the
number of vertices in the ternarized graph asymptotically equal to the
number of edges. The algorithm then runs a local procedure from every
vertex which discovers a subset of edges in the minimum spanning forest. We
show that after contracting the graph based on this discovered
fragment, the remaining graph has a factor of $n^{\epsilon/2}$ fewer
vertices.  At this point we can afford to invoke the dense routine from
Behnezhad \etal{} on the contracted graph which solves the problem in
$O((1/\epsilon) \log (1 /\epsilon))$ rounds.

The local procedure run at each vertex is to simply run Prim's
algorithm, a classic MSF algorithm, from the vertex until either of
the two following stopping conditions is met. Firstly, to ensure that
a machine does not perform too many queries, the algorithm truncates
the local search if it exceeds $n^{\epsilon}$ queries. Secondly,
running $n^{\epsilon}$ queries per vertex results may result in
performing $O(n^{1+\epsilon})$ queries in total, which would use
superlinear total space on sparse graphs. To restrict the total space
to nearly linear, we sample a random priority for every vertex, and
terminate the search if a vertex $v$'s Prim search visits a vertex $u$
with priority less than $v$. We argue that the overall query cost is
$O(m \log n)$ w.h.p. by relating the cost of each query to a certain
problem on randomized search trees (treaps).

\newcommand{\densemsf}[0]{\textsc{DenseMSF}}
\newcommand{\forestconn}[0]{\textsc{ForestConnectivity}}

We use the following algorithms from \cite{AMPC} as a
black-box.
\vspace{-0.1cm}
\begin{proposition}\label{prop:msfspaa}
  There is an \AMPC{} algorithm, \densemsf{}, which computes the
  minimum spanning forest of an undirected graph in
  $O((1/\epsilon)\log \log_{T/n} n)$ rounds w.h.p.  where the total
  space $T = \Omega(m + n)$.
\end{proposition}

\begin{proposition}\label{prop:forestconn}
  There exists an \AMPC{} algorithm, \forestconn{}, that solves the
  forest connectivity problem in $O(1/\epsilon)$ rounds of computation
  w.h.p.  using $T = O(n\log n)$ total space w.h.p.
\end{proposition}

\input{inputs/msfalgorithm}

\vspace{-0.6cm}

Note that this algorithm bears a close resemblance to the forest
connectivity algorithm of~\cite{AMPC}. One important
difference, however, is that the forest connectivity algorithm of
~\cite{AMPC} applies a routine which shrinks the number of vertices in
$G$ by a factor of $n^{\epsilon}$ iteratively for $O(1/\epsilon)$
rounds. The idea is to shrink the graph until the maximum size of a
tree is $O(n^{\epsilon/2})$. Unfortunately, a similar idea does not
work here, since although we show that the number of vertices in $G$
decreases by a factor of $n^{\epsilon}$ after one application of
Algorithm~\ref{alg:truncatedprim}, after contraction the graph is no
longer ternarized. Importantly, the number of edges in the contracted
graph may be asymptotically equivalent to the number of edges in the
original graph, which would result in the algorithm making no
progress. Instead, we use the fact that one application of
Algorithm~\ref{alg:truncatedprim} shrinks the number of vertices by a
factor of $n^{\epsilon}$, which is sufficient to apply the algorithm
from Proposition~\ref{prop:msfspaa} to the contracted graph.

\begin{restatable}{lemma}{msfshrink}\label{lem:msfshrink}
The contracted graph produced by Algorithm~\ref{alg:truncatedprim} has
a factor of $\Omega(n^{\epsilon/2})$ fewer vertices than $G$.
\end{restatable}

At a high level, to prove the lemma, we view $F_{\mu}$ as a collection of trees.
Then, we show that each vertex is a root of one tree with probability $O(n^{-\epsilon/2})$ (thanks to first stopping condition of Prim's algorithm in Algorithm~\ref{alg:truncatedprim}).
Hence, the number of trees in $F_{\mu}$ and consequently the number of vertices in the contracted graph shrinks by a factor of $\Omega(n^{\epsilon/2})$.

Next, we use the condition from Line~\ref{l:stopping} of Algorithm~\ref{alg:truncatedprim} to bound the total communication of the algorithm.

\begin{restatable}{lemma}{msfqueries}\label{lem:msfqueries}
Algorithm~\ref{alg:truncatedprim} uses $O(n \log n)$ queries w.h.p.
\end{restatable}

The following lemma argues that the vertices can be assigned to
machines in such a way that every machine performs $O(n^{\epsilon})$
queries per round w.h.p. The space-complexity bound follows directly
from our bound on the number of queries. Since the algorithm performs
a constant number of steps, each implementable in $O(1/\epsilon)$
rounds of \AMPC{}, the round-complexity follows.

\vspace{-0.1cm}
\begin{restatable}{lemma}{primspacerounds}\label{lem:primspacerounds}
  Algorithm~\ref{alg:truncatedprim} runs in $O(1/\epsilon)$ rounds and
  $O(n\log n)$ space w.h.p.
\end{restatable}
\vspace{-0.1cm}

By putting together the above lemmas, we obtain:

\begin{restatable}{lemma}{msfconstant}\label{lem:msfconstant}
Algorithm~\ref{alg:msf} computes the minimum spanning forest of an
undirected graph in $O(1/\epsilon \log(1/\epsilon))$ rounds w.h.p.
using total space $T = O(m\log n)$ w.h.p.
\end{restatable}

In order to prove Theorem~\ref{thm:mstlowspace} we need to improve the query complexity of the algorithm from $O(m \log n)$ to $O(m + n\log^2 n)$.
To that end, we use a sampling scheme by Karger, Klein and Tarjan~\cite{DBLP:journals/jacm/KargerKT95} used in the single-machine linear-time algorithm for computing minimum spanning tree.
By using basic algorithmic techniques on trees, i.e. finding lowest common ancestors and heavy-light decomposition, we show that the scheme can be implemented in a constant number of \AMPC{} rounds.

\subsection{Reducing the Query Complexity}
In this section we show the following reduction: given an AMPC algorithm for computing a minimum spanning forest in $O(1)$ rounds and makes $O(m \log n)$ queries in total, we can obtain an algorithm that runs in $O(1)$ rounds and makes $O(m + n \log^2 n)$ queries in total.
Hence, the query complexity of the algorithm is asymptotically optimal whenever $m = \Omega(n \log^2 n)$.

The reduction is obtained by combining a sampling lemma by Karger, Klein and Tarjan~\cite{DBLP:journals/jacm/KargerKT95} with basic algorithmic techniques on trees, i.e. finding lowest common ancestors and heavy-light decomposition.

\begin{definition}\label{def:flight}
Let $G = (V, E, w)$ be a weighted graph and $F$ be a forest, which is a subgraph of $G$.
For $x, y \in V$, let us define $w_F(x, y)$ as follows.
If $x$ and $y$ belong to different connected components of $F$, then $w_F(x, y) = \infty$.
Otherwise, we let $w_F(x, y)$ to be the maximum weight of an edge on the unique path from $x$ to $y$ in $F$.
We say that an edge $uw \in E$ is $F$-light, if $w(uw) \leq w_F(x, y)$, and $F$-heavy otherwise.
\end{definition}

By using basic properties from the minimum spanning tree, we get the following.

\begin{proposition}
Let $G = (V, E, w)$ be a weighted graph. Let $F$ be any forest of $G$ and $T$ be an arbitrary minimum spanning forest of $G$.
Then, all edges of $T$ are $F$-light.
\end{proposition}

Hence, it follows that when computing a minimum spanning forest of $G$ we can immediately discard all $F$-heavy edges.
To that end, we will use the following lemma.

\begin{lemma}{\cite{DBLP:journals/jacm/KargerKT95}}\label{lem:flight}
Let $G = (V, E, w)$ be an $n$-vertex weighted graph, let $H$ be a subgraph obtained from $G$ by including each edge independently with probability $p$, and let $F$ be the minimum spanning forest of $H$.
Then, the expected number of $F$-light edges in $G$ is $O(n/p)$.
\end{lemma}

The above lemma immediately suggests the following algorithm for computing a minimum spanning forest (see Algorithm~\ref{alg:msfreduction}).

\begin{tboxalg}{$MSF(G)$}\label{alg:msfreduction}
\begin{algorithmic}[1]
        \State $H := $ graph obtained from $G$ by sampling each edge independently with probability $1/\log n$.
	\State $F := $ compute the MSF of $H$

	\State \label{l:fl} $E_L := $ set of edges of $G$ which are $F$-light

	\State \Return the MSF of $F \cup E_L$.
\end{algorithmic}
\end{tboxalg}

\begin{lemma}\label{lem:remfast}
Algorithm~\ref{alg:msfreduction} is correct. All lines excluding~\ref{l:fl} can be implemented in $O(1)$ AMPC rounds using $O(m + n \log^2 n)$ total queries, where $m$ and $n$ are the numbers of edges and vertices of the input graph.
\end{lemma}

\begin{proof}
The correctness of the algorithm follows directly from Lemma~\ref{lem:flight}, combined with the fact that if $T$ is a minimum spanning forest of $G$, then $T$ is also a minimum spanning forest of every subgraph of $G$ containing $T$.

It remains to analyze the total number of rounds and queries that the algorithm makes.
Sampling graph $H$ clearly requires only $O(m)$ queries and $O(1)$ rounds.
To compute MSF of $H$ we use algorithm of Lemma~\ref{lem:msfconstant}, which requires $O((m/\log n) \log n + n) = O(m+n)$ total queries and $O(1)$ rounds.
Finally, the last step computes MSF of a graph that has $n$ vertices and $O(n \log n)$ edges (by Lemma~\ref{lem:flight}), which takes $O(n \log^2 n)$ queries and $O(1)$ rounds, again by using the algorithm of Lemma~\ref{lem:msfconstant}.
\end{proof}

Implementing line~\ref{l:fl} is quite technical and takes advantage of the fact that computing range-minimum queries, heavy-light decomposition and finding lowest common ancestors can all be done efficiently in the \MPC{} model.
We defer the full proofs to Appendix~\ref{apx:msf}. Combining
Lemma~\ref{lem:remfast} with with the MST result from
Lemma~\ref{lem:msfconstant} proves Theorem~\ref{thm:mstlowspace}.

%% file: inputs/msfalgorithm.tex
\newcommand{\msf}[0]{\textsc{MSF}}
\newcommand{\msfcontract}[0]{\textsc{MSFPrimContraction}}
\newcommand{\truncprimsearch}[0]{\textsc{MSFPrimContraction}}
\newcommand{\truncatedprim}[0]{\textsc{TruncatedPrim}}

\begin{tboxalg}{$\truncatedprim(G=(V, E))$, where $\Delta(G) \leq 3$}\label{alg:truncatedprim}
\begin{algorithmic}[1]
	\State For any vertex $v$ in the graph pick a rank $\pi(v)$ uniformly at random from $(0, 1)$, and let $\pi$ be the permutation obtained by sorting the vertices based on their ranks.
  \State Each vertex $v$ is assigned to a machine $\mu_v$ chosen
  uniformly at random from $m/n^\epsilon$ machines.

  \For{ each machine $\mu$ in parallel}
  \State $F_{\mu} = \emptyset$
    \For{ each $v$ s.t. $\mu_v = \mu$}
    \State Run Prim's algorithm at $v$, stopping either
    when\label{prim:conditions}
      \State\hspace{\algorithmicindent} (1) $v$ has explored $n^{\epsilon/2}$ vertices,
      \State\hspace{\algorithmicindent} (2) $v$'s component is fully explored, or
	\State\hspace{\algorithmicindent} (3) $v$ adds an edge to a vertex $u$ s.t. $\pi(u) < \pi(v)$.\label{l:stopping}
      \State $E_{v} \gets $ MST edges discovered by $v$. \Comment{emitted as part of the output}
      \If {$v$ stops due to case (3)}
        \State $F_{\mu} \gets F_{\mu} \cup \{(v, u)\}$\label{prim:addedge}
      \EndIf
    \EndFor
  \EndFor

  \State Apply the algorithm from Proposition~\ref{prop:forestconn} to
  $F = \cup_{\mu} F_{\mu}$, which contracts directed trees in $F$ to
  their roots, and let $C : V \rightarrow V$ be a mapping representing
  the contraction.\label{prim:forestconn}

  \State Let $G'(V',E')$ be the graph obtained by contracting $G$
  according to $C$, with isolated vertices
  removed.\label{prim:contract}

  \State \Return $(\cup_{v \in V} E_{v}, G')$
\end{algorithmic}
\end{tboxalg}

\begin{tboxalg}{$\msf(G=(V, E))$}\label{alg:msf}
\begin{algorithmic}[1]
  \If{ $m < n^{1 + \epsilon/2}$}
     \State Let $G'(V', E')$ be a degree bounded version of $G$,
       obtained by replacing every vertex $v$ with degree $> 3$ with a
       cycle of length $\emph{deg}(v)$, connecting each edge of $v$
       to its corresponding vertex in the cycle. Let the weights of
       the dummy edges be denoted by $\bot$, chosen to be less than
       the weight of the lightest edge in $E$.\label{msf:degbound}

     \State $(G''(V'',E''), E_{T}) \gets \truncatedprim{}(G')$. Note
     that $|V''| = O(m^{1-\epsilon/2})$.\label{msf:runtruncprim}

    \State $E'_{T} \gets $ edges obtained from applying
    the algorithm from Proposition~\ref{prop:msfspaa} to $G''$.\label{msf:runspaa}

    \State \Return $E_{T} \cup E'_{T}$, with all edges with weight
        $\bot$ removed.\label{msf:removebot}
  \EndIf
  \State \Return edges obtained from applying the algorithm from
  Proposition~\ref{prop:msfspaa} to $G$\label{msf:rundense}
\end{algorithmic}
\end{tboxalg}

%% file: inputs/matching.tex
\section{Matching}\label{sec:matching}

\newcommand{\lfmm}[1]{\ensuremath{\mathsf{GreedyMM}(#1)}}

\begin{tboxalg}{Algorithm for maximal matching.}\label{alg:maximalmatching}
	\begin{algorithmic}[1]
	\State For any edge $e$ in the graph, pick a rank $\pi(e)$ uniformly at random from $(0, 1)$ and let $\pi$ be the permutation obtained by sorting the edges based on their ranks.
	\State $G_1 \gets G$
	\For{$i \in 1 \ldots k=\lceil \log_2 \log_2 \Delta \rceil + 1$}
		\If{$\Delta(G_i) > 10\log n$}\label{line:if}
			\State Let $H_i$ be the subgraph of $G_i$ containing its edge $e$ iff $\pi(e) \in [0, \Delta^{-0.5^i}]$.
		\Else
			\State $H_i \gets G_i$.
		\EndIf
		\State Find matching $M_i = \lfmm{H_i, \pi}$ by running the MIS algorithm of Proposition~\ref{prop:mis} on the line graph of $H_i$, using permutation $\pi$.
		\State $G_{i+1} \gets G_i[V \setminus V(M_i)]$.
	\EndFor
	\State \Return matching $M_1 \cup M_2 \cup \ldots \cup M_k$.
\end{algorithmic}
\end{tboxalg}

Maximum matching and its natural extension maximum weight matching are among the most fundamental combinatorial optimization problems with a wide range of applications. Particularly, maximum weight matching is an important subroutine in balanced partitioning and hierarchical clustering, see e.g. \cite{icmlVahab} and the references therein.

In this section, we consider the matching problem in the \AMPC{} model. Our main result is an efficient algorithm for the unweighted {\em maximal matching} problem. We then use this algorithm as a black-box to get an algorithm for approximate maximum weight matching and related problems.

\begin{theorem}\label{thm:maximalmatching}
There is an \AMPC{} algorithm which with probability $1-1/\poly(n)$ computes a random greedy maximal matching using $O(n^\delta)$ space per machine (for any constant $\delta \in (0, 1)$) and:
\begin{enumerate}[itemsep=0pt,topsep=0pt]
	\item In $O(\log \log n)$ rounds using $\widetilde{O}(m + n)$ total
  space.\label{lab:part1}
	\item In $O(1)$ rounds using $O(m+n^{1+\epsilon})$ total space for any constant $\epsilon > 0$.
  \label{lab:part2}
\end{enumerate}
\end{theorem}

The \emph{random greedy maximal matching} is obtained by considering the edges one by one in a random order and adding each edge to the matching if none of its endpoints has been previously matched.
Theorem~\ref{thm:maximalmatching} leads to the following results.

\begin{corollary}
The same guarantee as in Theorem~\ref{thm:maximalmatching} also applies to $1+\epsilon$ approximate maximum matching, $2+\epsilon$ approximate maximum weight matching, and $2$ approximate minimum vertex cover.
\end{corollary}

We use the following MIS algorithm of \cite{AMPC} as a black-box.

\begin{proposition}{\cite{AMPC}}\label{prop:mis}
There is a randomized \AMPC{} algorithm which computes a random-greedy maximal independent set in $O(1)$ rounds using $O(n^\delta)$ space per machine (for any constant $\delta \in (0, 1)$) and $\widetilde{O}(m)$ total space. The bound on the total space holds in expectation and the rest of the bounds hold with high probability.
\end{proposition}

Our algorithm uses the following relation between maximal independent set and maximal matching.
Given an undirected graph $G$, the \emph{line graph} of $G$ is obtained by having a vertex for each edge of $G$ and connecting any two of these vertices that correspond to edges sharing an endpoint in $G$.
It is well-known that the set of vertices in the maximal independent set of the line graph of a graph $G$ forms a maximal matching of $G$.  Therefore, having Proposition~\ref{prop:mis} which gives an efficient \AMPC{} MIS algorithm, one may hope to be able to directly get an efficient algorithm for maximal matching as well. Unfortunately, this is not the case. The main hurdle is that the line graph can be significantly larger than the graph itself and thus the total space of an algorithm constructing the line graph may be super-linear in the number of edges of the original graph. To get around this, we take two different approaches leading to the two bounds obtained in Theorem~\ref{thm:maximalmatching}.

For the first bound, our main idea is to run the MIS algorithm on the line graph of a smaller edge-sampled subgraph of the original graph $G$, commit the edges of the obtained (not maximal) matching to the output (which results in pruning this graph) and then recursively handle the residual graph. We show that $O(\log \log \Delta)$ iterations suffice to find a maximal matching, where $\Delta$ denotes the maximum degree of the input graph, while also ensuring that the total required space remains $\smash{\widetilde{O}(m)}$.

For the second bound, first notice that the line-graph of a graph with $m$ edges and maximum degree $\Delta$ may have up to $m \Delta$ edges. Hence, if we explicitly construct the line-graph and run the MIS algorithm as a black-box we may require $\Omega(m \Delta)$ total space. We note that a simple adaption of the idea used in the algorithm of Proposition~\ref{prop:mis} can be used to reduce the total space to $O(m^{1+\epsilon})$. The idea is to run the algorithm without explicitly constructing the line-graph a-priori, but rather generating it on the fly. To improve this further to $O(m + n^{1+\epsilon})$, we use the structure of line-graphs and use a more efficient query process for matchings as opposed to the one (due to \cite{yoshida}) used in Proposition~\ref{prop:mis} for maximal independent sets.

\subsection{Theorem~\ref{thm:maximalmatching} (Part~\ref{lab:part1})}\label{sec:thmmatchingpart1}

\smparagraph{Notation.} We need a few definitions to formalize the algorithm. Given a graph $G(V, E)$ and a subset $V'$ of $V$, we use $G[V']$ to denote the induced subgraph of $G$ on vertex set $V'$. Furthermore, given a matching $M$ in $G$, we use $V(M)$ to denote the vertices matched in $M$. Also, given a permutation $\pi$ over the edges in $E$, we denote by $\lfmm{G, \pi}$ the greedy maximal matching obtained by iterating over the edges in the order of $\pi$.
Finally, $\Delta(G)$ denotes the maximum vertex degree in $G$.
The algorithm is formalized as Algorithm~\ref{alg:maximalmatching}. It is easy to verify that each iteration of the for loop can be implemented in $O(1)$ rounds of \AMPC{} so long as the space per machine is $n^{\Omega(1)}$. Therefore, the number of rounds is indeed $O(\log \log \Delta)$. In the remainder of this section, we prove the space bounds and the algorithm's correctness.
We need the following well-known degree-reduction property of lexicographically first maximal matching in our analysis. See for instance {\cite[Lemma A.1]{DBLP:journals/corr/abs-1901-03744}}.

\begin{proposition}\label{prop:degreereduction}
	Let $G$ be any undirected graph, $\pi$ be a random permutation on the edges of $G$, and $p$ be a parameter in $(0, 1)$. The maximum degree in graph $G[V \setminus V(\lfmm{G_p, \pi})]$ is w.h.p. $5\log n / p$ where $G_p$ is the subgraph of $G$ including only the $p$ fraction of edges with the lowest rank in $\pi$.
\end{proposition}

Using the above proposition, we obtain the following.

\begin{restatable}{lemma}{degreereduction}\label{lem:degreereduction}
	For any $i$, the maximum degree in graph $G_i$ is at most $5\Delta^{0.5^{i-1}}\log n$ w.h.p.
\end{restatable}

\begin{proof}
	One can verify from the algorithm's description that graph $G_i$ is indeed graph $G[V \setminus V(M_1 \cup \ldots \cup M_{i-1})]$ and can also verify that $M_1 \cup \ldots \cup M_{i-1}$ is the LFMM of the edges $e$ in $G$ with $\pi(e) \leq \Delta^{-0.5^{i-1}}$. Therefore by Proposition~\ref{prop:degreereduction}, the maximum degree in graph $G_i$ is w.h.p. bounded by $5\log n / \Delta^{-0.5^{i-1}} = 5\Delta^{0.5^{i-1}}\log n$.
\end{proof}

\begin{restatable}{lemma}{matchingcorrectness}
Algorithm~\ref{alg:maximalmatching} is correct and uses $\tilde{O}(m)$ space.
\end{restatable}
\begin{proof}
	Since $G_{i+1} = G_i[V \setminus V(M_i)]$, all the matchings $M_1, \ldots, M_k$ will be vertex disjoint and thus their union is a valid matching of the graph. It remains to show that this is a maximal matching. Recall from Lemma~\ref{lem:degreereduction} that the maximum degree in graph $G_i$ is $5\Delta^{0.5^i}\log n$. Therefore, for $k' = \lceil \log_2\log_2 \Delta \rceil$, we can bound the maximum degree in $G_{k'}$ by
	$$
		5\Delta^{0.5^{\lceil\log_2\log_2 \Delta\rceil}}\log n \leq 5\Delta^{\frac{1}{\log_2 \Delta}}\log n = 10\log n.
	$$
	This means that in the next iteration $k = k'+1$, the condition of
  Line~\ref{line:if} does not hold and we have $H_k = G_k$ and thus
  the maximal matching $M_k$ of $H_k$ is also a maximal matching of
  $G_k$. Thus, $M_1 \cup \ldots \cup M_k$ is a maximal matching of $G$.
\end{proof}

\subsection{Theorem~\ref{thm:maximalmatching} (Part~\ref{lab:part2})}\label{sec:thmmatchingpart2}

Suppose that we are given a random permutation $\pi$ over the edges, and are tasked to determine whether an edge $e$ is matched in the corresponding random greedy matching $M$. The original query process of \cite{yoshida} is as follows for matchings: We iteratively pick the incident edge $f$ to $e$ with the lowest rank in $\pi$. If $f$ happens to have a higher rank than $e$, that is, if no edge of lower rank than $e$ is incident to $e$, then $e$ must be in the matching. Otherwise, we recursively query $f$. If $f$ happens to be part of the matching, $e$ is not and we can terminate the process; if $f$ is not in the matching, we are unsure about the status of $e$ and proceed to the next incident edge to $e$.

A truncated variant of this process proposed in \cite{AMPC} follows
the same recursive idea, but truncates it if the total number of recursive queries exceeds $n^\epsilon$. A similar analysis as done in \cite{AMPC} for MIS which builds on that of \cite{yoshida} shows that if we run this truncated query process on all the edges, the number of edges whose query process is truncated is small enough that significant progress on the whole graph is made. Therefore, by repeating this process $O(1)$ times, we get a maximal matching. Particularly it is shown in \cite{AMPC} that:

\begin{lemma}\label{lem:misampc}
	Suppose that we run the $n^\epsilon$-truncated query process above, and remove all edges that are known to join the matching along with their incident edges. Then $O(1/\epsilon)$ applications of this process makes the graph empty.
\end{lemma}

The problem with this approach is that since we conduct
$n^\epsilon$ queries from each edge, the total number of queries (and
equivalently the total communication/space) can be $\Theta(m
n^\epsilon)$.

To reduce the total space to $O(m + n^{1+\epsilon})$, instead of edges, we start the query processes from the {\em vertices} and truncate them. The query process of a vertex $v$, basically iterates over the incident edges of $v$ in the increasing order of their ranks and runs the edge query process. Once an incident edge in the matching is found, the vertex query process terminates.

The truncated variant of the vertex query process above is also
natural: Once the total number of recursive queries for the vertex
exceeds $n^\epsilon$, we stop and mark the vertex as unsettled. It is
clear that since there are $n$ vertices and each queries only
$n^\epsilon$ portion of the graph, the total space needed is only $O(m
+ n^{1+\epsilon})$. In the full version, we show that this
vertex-truncated process also makes enough progress on the graph that
after $O(1)$ rounds, we find a maximal matching.

\begin{lemma}
	Suppose that we run the $n^\epsilon$-truncated query process above from the vertices. Then we remove all vertices known to be matched along with their incident edges. Then $O(1/\epsilon)$ applications of this process makes the graph empty.
\end{lemma}

\begin{proof}[Proof sketch]
	Let us denote the query size of the edge process on some edge $e$ by $q(e)$ and the query size of the vertex process on some vertex $v$ by $q(v)$.

	Consider an edge $e$. If its edge-query process terminates within $n^\epsilon$, then either $e$ joins the matching or there is at least one of its neighbors $f$ that joins it. In the latter case let $v = f \cap e$ and for the former let $v$ be any endpoint of $e$. In either case, one can confirm that
	$
		q(v) \leq q(e).
	$
	Since in either case, vertex $v$ joins the matching, even in the vertex-query process, we detect that $v \in M$ and remove edge $e$. Therefore, any edge $e$ that is removed in one iteration of the edge-query process, is also removed in one iteration of the vertex-query process. Thus Lemma~\ref{lem:misampc} can be used to finish the proof.
\end{proof}

%% file: inputs/experiments_new.tex
\section{Empirical Evaluation}\label{sec:experiments}
In this section we provide an empirical evaluation of the \AMPC{}
model using hundreds of threads on production machines in a large
data center.
We summarize the main experimental results described in this section
below:

\begin{itemize}[topsep=0pt,itemsep=0pt,parsep=0pt,leftmargin=15pt]
  \item \revised{A case-study of \AMPC{} and \MPC{} implementations
  of MIS (Section~\ref{sec:miscasestudy}), including a discussion of
  different optimizations, round-complexity, communication and a
  detailed evaluation explaining their performance.}

  \item \revised{\AMPC{} implementations of Maximal Matching
  (Section~\ref{sec:mmexps}) and Minimum Spanning Forest
  (Section~\ref{sec:msfexps}), detailed experimental comparisons
  explaining their performance, and a comparison to the state-of-the-art
  \MPC{} algorithms that we implemented in this paper.}

  \item \revised{Experimental evaluation of the 1-vs-2 Cycle problem,
  comparing \AMPC{} vs \MPC{} implementations on synthetic cycle
  graphs (Section~\ref{sec:twocycleexps}).}

  \item \revised{A discussion studying potential scalability
  bottlenecks in our algorithms, as well as the potential for applying
  our approach to a broader set of problems
  (Section~\ref{sec:discussion}).}

\end{itemize}

\input{inputs/experimental_environment}
\input{inputs/graph_sizes}

\input{inputs/graph_inputs}
\input{inputs/mis_case_study}

\input{inputs/mm_experiments}

\input{inputs/msf_experiments}
\input{inputs/cycle_experiments}

\input{inputs/discussion}

%% file: inputs/experimental_environment.tex
\subsection{Environment \& Implementation}
We first provide context and rationale for the implementations and
environment utilized in our evaluation.

\shep{
\myparagraph{Distributed Environment}
In this paper we focus on a \emph{fault-tolerant distributed
environment} which reflects the challenges of running large-scale
graph computations in shared production data center.
The setting used in our experiments has recently been
described in Tirmazi~\etal{}~\cite{tirmazi2020borg}.
In this setting, batch jobs are typically run at low priorities
(i.e., using resources that are currently not used by high priority jobs),
which makes them susceptible to preemptions.
While running batch jobs at higher priorities is possible, it is
also much more costly~\cite{aws-pricing, gcp-pricing} and thus often not done in practice.
This is why systems like MapReduce, Hadoop or Flume-C++~\cite{flumecpp}
have strong fault tolerance properties and write the results of each computation
round to durable storage.

The focus of our empirical evaluation is to study the performance of \AMPC{} algorithms, and compare them with state-of-the-art \MPC{} algorithms using a distributed computation framework with good fault-tolerance properties.

}

\revised{
\myparagraph{Implementing \AMPC{} and \MPC{} Algorithms}
We implement both the \AMPC{} and \MPC{} algorithms studied in this
paper in C++ using Flume-C++~\cite{flumecpp,flumejava}.
\shep{Flume-C++ is a highly-optimized fault-tolerant parallel data
processing framework, whose API is similar to the open source system, Apache Beam.}
A Flume-C++ implementation consists of \emph{stages}, that consume
inputs generated by previous stages, and emit outputs that can be
consumed by later stages. Flume automatically handles performing
fusion, and inter-stage optimizations.
\shep{Compared to other popular data processing frameworks, such as
Spark~\cite{spark} and Timely Dataflow~\cite{murray2013naiad}, the
main difference between Flume-C++ and these systems is that the
Flume-C++ runtime ensures that workers write the outputs of stages in
the computation to durable storage for the duration of the job. This
logging enables fault-tolerance, which is critical in the shared data
center we run our jobs in.}

To implement the \AMPC{} algorithms, we use a custom distributed
key-value store optimized for lookups and high throughput.  The
implementation used in our experiments takes advantage of hardware
support for Remote Direct Memory Access (RDMA), which is a widely
available technology~\cite{infiniband, roce, omnipath}.  \shep{We also
experiment with an alternate implementation of the key-value store
only uses TCP/IP, but unless explicitly specified, our experiments
use the RDMA-based implementation.}
The Flume-C++ implementations in this paper amount to a few hundred
lines of C++ each. We emphasize that the \emph{only difference}
between our \AMPC{} and \MPC{} implementations is that the \AMPC{}
codes can query a key-value store within a Flume stage.
}

\myparagraph{Machine Configuration}
Our experiments are performed on machines from a production
datacenter. Each machine contains two 2-way hyper-threaded 2.3GHz Intel 18-core CPUs, for a total of 72 hyper-threads per machine.
Each machine is equipped with 262GB of RAM,
and a 20Gbps NIC.
We use a maximum of 100 machines to solve all problems, although we
note that in practice we use far fewer than $72 \times 100$
hyper-threads (we request 400 hyper-threads but the experiments may
use slightly more
depending on cluster availability). Our experiments
are run on a shared cluster, and may thus compete for resources with
other jobs, and run virtually alongside other jobs on the same
machine. \shep{To mitigate these effects, we ran all of our
experiments using relatively high priorities.
Thus, although we run our jobs in a fault-tolerant environment, our
jobs typically do not experience failures.
As a further precaution, we also run our experiments 3 times and
report the median running time. We note that the difference in running
times across different trials was not significant (within 10\%).}

%% file: inputs/graph_sizes.tex
\setlength{\tabcolsep}{1.5pt}
\begin{table}[!t]
\small
  \centering
\begin{tabular}{l|r|r|r|r|r}
\toprule
{\bf Dataset}                    & {\bf $n$} & {\bf $m$} & {\bf Diam.} & {\bf Num. CC} & {\bf Largest CC} \\
        \midrule
{$2 \times k$}             & $2\times k$    & $2 \times k$         & $k$                       & 2            & $k$                       \\ 
{OK}             & 3.07M     & 234.4M                     & 9                         & 1            & 3.1M                       \\ 
	{\revised{TW}}           & \revised{41.6M}     & \revised{2.4B}                       & \revised{23*}                       & \revised{2}            &\revised{41.6M}                       \\ 
{FS}        & 65.6M     & 3.6B                       & 32                        & 1            &65.6M                       \\ 
{CW}           & 0.978B    & 74.7B                      & 132*                      & 23,794,336   &0.950M                      \\ 
{HL}     & 3.56B     & 225.8B                     & 331*                      & 144,628,744  &3.35B                       \\ 
\end{tabular}
\caption{Graph inputs, including vertices and edges, diameter, the
number of components, and the size of the largest component. We mark
diameter values where we are unable to calculate the exact diameter
with * and report a lower bound on the diameter observed in prior
experimental work \protect\cite{dhulipala2018theoretically}.}
\label{table:sizes}`
\end{table}

%% file: inputs/graph_inputs.tex
\subsection{Graph Data}
We evaluate our algorithms on a representative set of real-world
graphs of varying sizes.
Several of our datasets are from the SNAP network suite:
\defn{com-Orkut (OK)} is an undirected graph of the Orkut social
network, and \defn{Friendster (FS)} is an undirected graph of the
Friendster social network.  \revised{\defn{Twitter (TW)} is a graph of
the Twitter network, where edges represent the follower
relationship~\cite{kwak2010twitter}.} Lastly, we use the \defn{ClueWeb
(CW)} graph, which is a Web graph from the Lemur project at CMU and
Google, obtained from the LAW collection of
datasets~\cite{boldi2004webgraph}. \defn{Hyperlink2012 (HL)} is a
directed hyperlink graph obtained from the WebDataCommons dataset
where nodes represent web pages~\cite{meusel15hyperlink}.
\revised{The Twitter, ClueWeb, and Hyperlink2012 graphs are originally
directed graphs, so we symmetrize them before applying our
algorithms.} We test our 2-cycle algorithms on a family of massive
high-diameter graphs consisting of two cycles on $k$ vertices each
($2\times k$ graphs) \revised{Finally, we test our minimum spanning
forest algorithm on the same graph inputs where the weight of an edge
$(u,v)$ is proportional to $\deg(u) + \deg(v)$.}

%% file: inputs/mis_case_study.tex
\revised{
\subsection{Case Study: Maximal Independent Set}\label{sec:miscasestudy}

To provide insights into how our implementations are programmed, in
this subsection we give a detailed description of our Flume-C++
\AMPC{} and \MPC{} implementations of Maximal Independent Set (MIS)
algorithms.

\input{inputs/fig_ampc_mis}

\myparagraph{\AMPC{} Algorithm}
Figure~\ref{fig:ampc_mis} provides high-level pseudocode for the
\AMPC{} MIS algorithm studied in this paper. In our pseudocodes,
we use several concepts
from~\cite{flumejava} (also used in Apache Beam), which we now describe. A
\textbf{\textsc{PCollection}} is a potentially distributed,
multi-element data set.
A \textbf{\textsc{KV<T, S>}} is a key-value pair, whose key has type T
and value has type S.
A \textbf{\textsc{DoFn<T, S>}} is an operation that transforms a
\textsc{PCollection<T>} to a \textsc{PCollection<S>}. The
graph input is a \textsc{PCollection} of \textsc{KV} mapping
\textsc{NodeId}s to
\textsc{Node}s.  A \textbf{\textsc{Node}} is a list of
\textsc{NodeId}s of the neighbors. Thus, every key-value pair in the
graph represents a single vertex.

Our implementation is based on the $O(1)$ round \AMPC{} algorithm
described by Behnezhad~\etal{}~\cite{AMPC}. The algorithm finds the
{\em lexicographically first MIS} over a random ordering of the
vertices.
For convenience, we assume that each vertex is given a random priority
defining its rank in the permutation.

Behnezhad~\etal{}~\cite{AMPC} show that the recursive algorithm of
Yoshida~\etal{}~\cite{yoshida} can be adapted to run in $O(1)$ \AMPC{}
rounds.  The algorithm of Yoshida~\etal{} is based on the following
recursive idea: each vertex belongs to the MIS if and only if none of
its lower priority neighbors belongs to the MIS.  Somewhat
surprisingly, as shown in~\cite{yoshida}, the natural way of turning
this property into a recursive function gives an algorithm that runs
in linear total time ($O(m)$), even if the recursion is run separately
from each vertex of the graph and with no memoization.

The \AMPC{} algorithm of~\cite{AMPC} runs this recursion in
$O(1/\epsilon)$ steps.  After step $i$ all vertices that needed to
make $O(n^{i\cdot \epsilon})$ recursive calls learned whether they
belong to the MIS.  This multi-stage approach turned out not to be
needed in practice, and our \AMPC{} implementation only needs $2$
rounds of computation to find the MIS.  See Figure~\ref{fig:ampc_mis}
for the pseudocode.

\input{inputs/fig_mpc_mis}

\myparagraph{\MPC{} Algorithm}
Our \MPC{} algorithm is a recent $O(\log n)$-round
\emph{rootset-based} algorithm for computing the
lexicographically-first MIS with respect to a random permutation on
the vertices~\cite{DBLP:conf/spaa/BlellochFS12}.
This algorithm was recently shown to have a $O(\log n)$-round
complexity by Fischer and Noever~\cite{fischer2018mis}.
By specifying the same source of randomness, both the \MPC{} and
\AMPC{} algorithms compute the same MIS.

Conceptually, the rootset-based algorithm begins by drawing a random
number for each vertex in the graph. Then, it proceeds in phases.
In each phase, it finds all vertices whose priority is lower than the
priority of all their neighbors.
All such vertices are added to the MIS, and after that they are removed
from the graph together with their neighbors.
After $O(\log n)$ phases, all vertices are removed from the
graph~\cite{fischer2018mis}.
The pseudocode of our MPC implementation is given in Figure~\ref{fig:mpc_mis}.
We note that although a $O(\sqrt{\log n})$-round \MPC{} algorithm is also known~\cite{DBLP:conf/soda/GhaffariU19}, the complexity of the algorithm makes it likely impractical.
Finally, we experimentally determined that switching to an in-memory
algorithm once the number of edges in the graphs decreases below $5
\times 10^{7}$ achieves a good tradeoff between the cost of processing
the graph on a single machine, and the cost of a new phase.

\shep{
We note that we also considered an \MPC{} implementation of the
\AMPC{} algorithm as a potential baseline, in which each step of
querying the key-value store was mapped to a shuffle.
We observed that this algorithm requires over 1000 shuffles even
for the Orkut and Friendster graphs, and is over 50x slower than
the rootset-based algorithm, and thus we chose to use the
rootset-based algorithm as our \MPC{} baseline.
}

\setlength{\tabcolsep}{3pt}
\begin{table}[!t]
\small
  \centering
%\scriptsize
\begin{tabular}{l|c|c|c|c|c}
\toprule
{\bf Algorithm}                    & {\bf OK} & {\bf TW} & {\bf FS} & {\bf CW} & {\bf HL} \\ % & {\bf $2 \times 10^{8}$} & {\bf $2 \times 10^{9}$} &{\bf $2 \times 10^{10}$} \\
        \midrule
{\AMPC{} Maximal Independent Set }  & 1 & 1 & 1 & 1 & 1 \\
{\AMPC{} Maximal Matching }  & 1 & 1 & 1 & 1 & 1 \\
{\AMPC{} Minimum Spanning Forest } & 5 & 5 & 5 & 5 & 5 \\
\midrule
{\MPC{} Maximal Independent Set }  & 8 & 10 & 10 & 12 & 14 \\
{\MPC{} Maximal Matching}   & 8 & 12 & 12 & 14 & 16 \\
{\MPC{} Minimum Spanning Forest}  & 33 & 54 & 57 & 84 & -- \\
\end{tabular}
\caption{\revised{Number of shuffles (costly rounds) used by our \AMPC{} and
\MPC{} implementations for real-world graph datasets.}}
\label{table:shuffles}
\end{table}

\begin{figure}[!t]
\begin{center}
\includegraphics[scale=0.55]{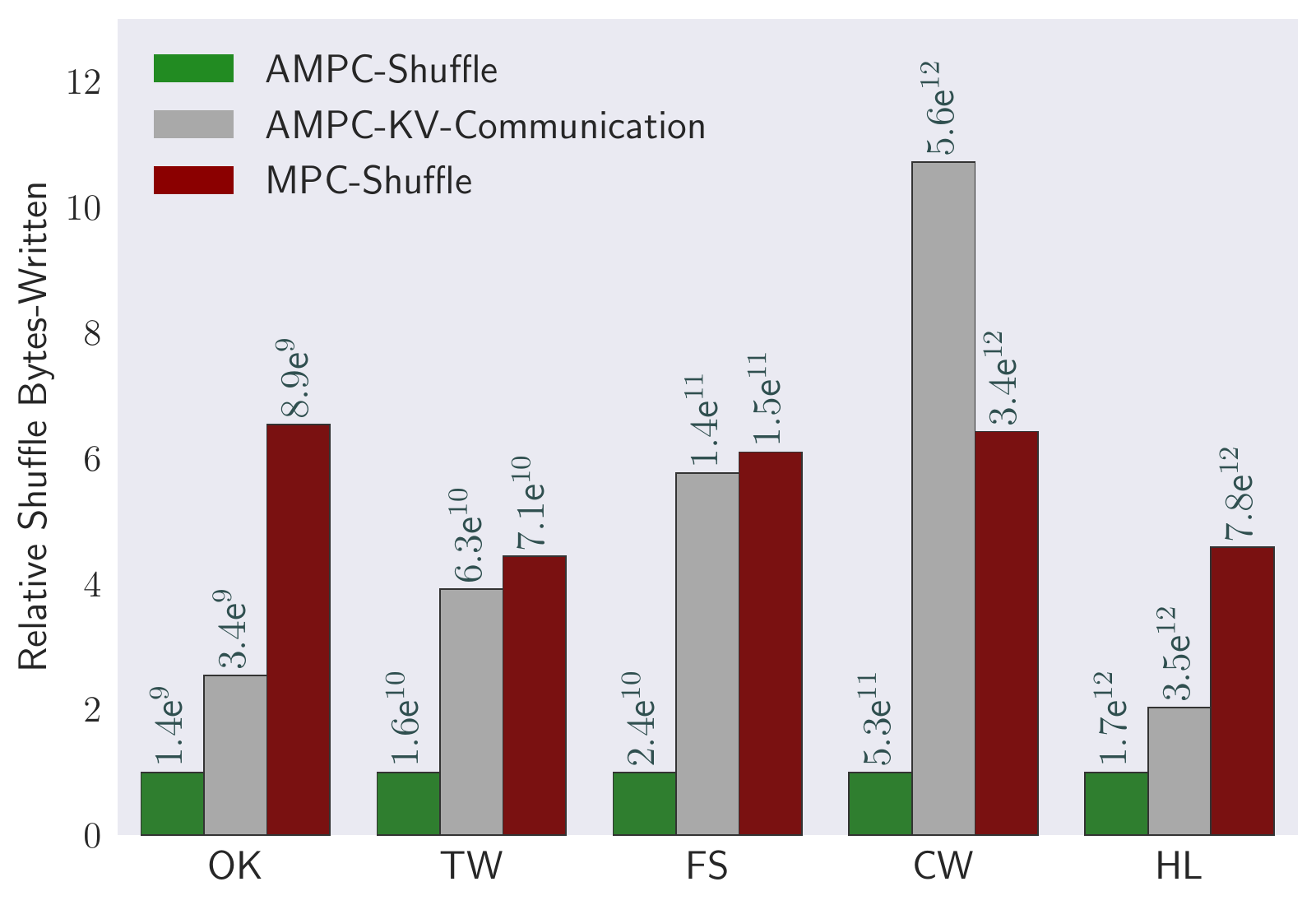}
\caption{\label{fig:mis_communication}
\shep{
Normalized bytes shuffled for the \AMPC{} and \MPC{} MIS
algorithms, and the normalized bytes of communication with the
key-value store for the \AMPC{} MIS algorithm. Each bar is annotated
with the actual number of bytes on top of the bar.
}
}
\end{center}
\end{figure}

\myparagraph{Round-Complexity and Communication}
What is the empirical benefit of reducing the round-complexity? For
one, the implementations of our \AMPC{} algorithms perform fewer \emph{shuffles}. A \defn{shuffle} is a phase which takes a set of emitted key-value pairs
and groups pairs with the same key on the same
machine~\cite{flumejava} and is the only way a
Flume-C++ worker can exchange big
amount of data with other workers.
At a high level, a single Flume-C++ shuffle corresponds to a single round in the \MPC{} and \AMPC{} models.
We point out
that empirically, most of the computation time in the \MPC{}
algorithms studied in this paper is spent on shuffles.  In
Table~\ref{table:shuffles} we report the number of shuffles (or \MPC{}/\AMPC{} rounds) used by our \AMPC{} MIS algorithm and our \MPC{} MIS algorithm. The \AMPC{}
algorithm uses a single shuffle to construct a directed graph (step
(1) in Figure~\ref{fig:ampc_mis}). The \MPC{} algorithm requires two
shuffles per phase of the algorithm, and takes between 8--14 shuffles.

Although the graph size shrinks between each phase, each shuffle still
transmits a significant amount of data.
To illustrate the impact of performing fewer shuffles in the \AMPC{}
algorithm, in Figure~\ref{fig:mis_communication} we plot the total
number bytes written during all shuffles in the algorithm for our
\AMPC{} and \MPC{} algorithm, as well as the total amount of
communication performed by the \AMPC{} algorithm to the key-value
store. In all cases, the \AMPC{} algorithm shuffles significantly
fewer bytes, since the single shuffle it performs writes bytes only
proportional to the input graph size.
Note that the total communication to the key-value store in the
\AMPC{} algorithm is typically less than the total bytes shuffled by
the \MPC{} algorithm, with the exception of the ClueWeb graph.
Despite the fact that the \AMPC{} algorithm sometimes uses more total
bytes of communication, it is always faster as we discuss later, since
the key-value store communication is done over a relatively
high-throughput network.

\myparagraph{Optimizations in the \AMPC{} Algorithm}
Next, we consider two optimizations that can be applied to the \AMPC{}
algorithm to improve its performance.

The first is a \defn{multithreading} optimization that is broadly
applicable to any \AMPC{} algorithm that performs lookups to a
key-value store. Since in our programming model, these lookups are
performed synchronously, we use multiple threads to enable each worker
to concurrently process many instances of a \textsc{DoFn}. This
optimization enables a thread processing a task that is waiting to
receive a result key-value store to be swapped out for another thread.
Since RDMA lookups to the key-value store are in general an order of
magnitude slower than lookups to DRAM, enabling this optimization
should improve the running time of the part of the algorithm that
performs lookups to the key-value store.

The second is a \defn{caching} optimization that we apply to the MIS
and Maximal Matching algorithms in this paper. The idea is essentially
to save the result of a query determining whether a given vertex (or
edge) is preserved in the MIS. Theoretically, since each machine only
performs $O(n^{\epsilon})$ communication, the results of all queries
it recursively answers can be saved on this machine in the model. In
practice, we implement the caching optimization using an array indexed
over the vertices that is shared between all threads operating on a
machine. In our MIS algorithm, this table stores a three-valued state
reporting whether the status of this vertex is either \emph{Unknown},
\emph{InMIS} or \emph{NotInMIS}.

\begin{figure}[!t]
\begin{center}
\includegraphics[scale=0.55]{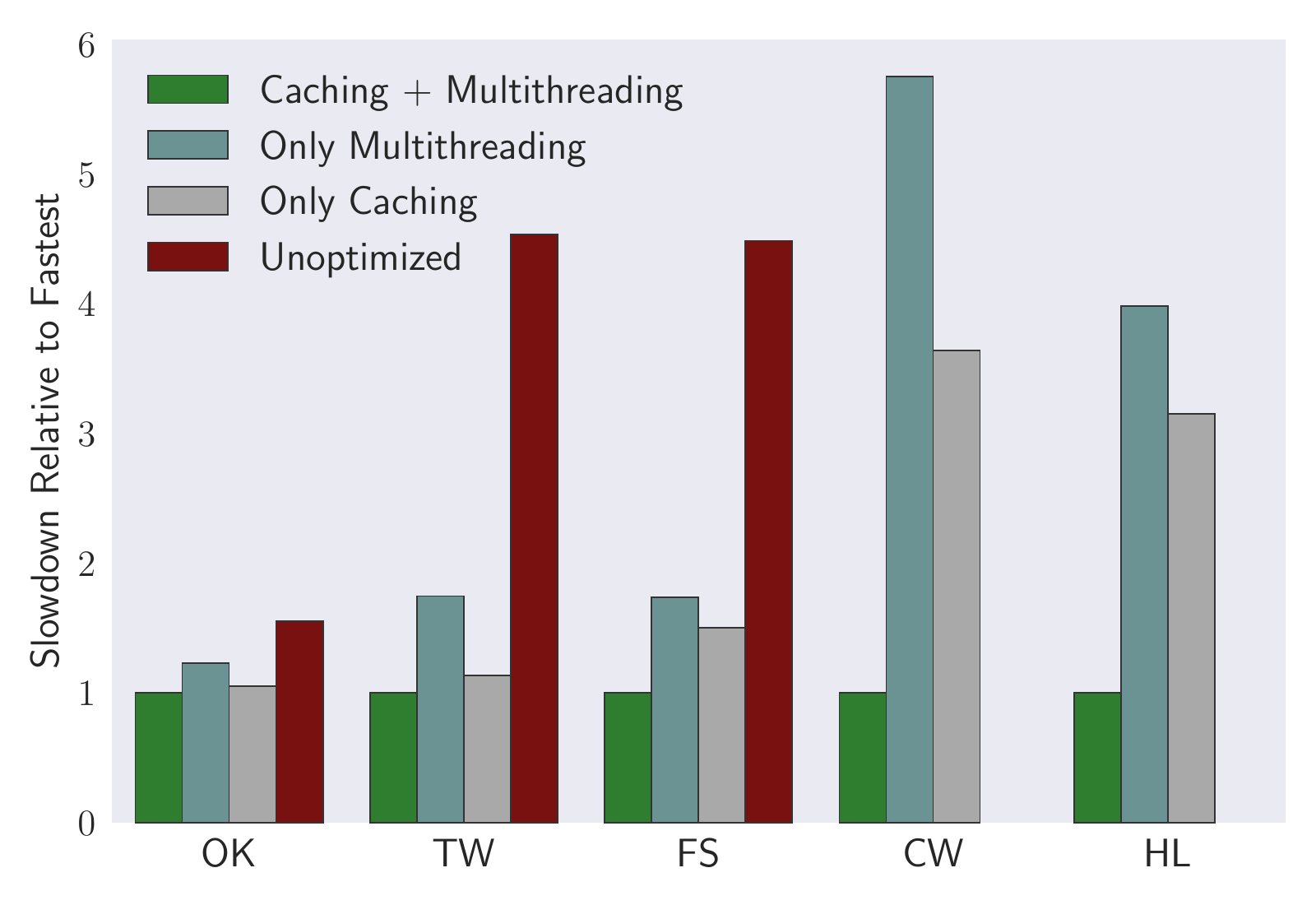}
\caption{\label{fig:mis_optimizations}
Effect of the caching and multithreading optimizations for our
\AMPC{} Maximal Independent Set implementation. The \AMPC{} MIS
algorithm without optimizations did not finish within 4 hours for both
the CW and HL graphs.
}
\end{center}
\end{figure}

Figure~\ref{fig:mis_optimizations} shows the impact of the caching
and multithreading optimizations for our \AMPC{} MIS implementation
on the overall running time.
We note that only the running time of Step (3) in our algorithm in
Figure~\ref{fig:ampc_mis} is affected by these optimizations (we tried
enabling multithreading for the other steps, but the results were
always slower).
Caching decreases the running time of this step by reducing the number
of bytes that must be communicated with the key-value store.
We observe that both optimizations always provide speedups over a
baseline, and that the fastest times are obtained when both
optimizations are applied.
Only using multithreading enables a 1.26--2.59x speedup over the
unoptimized algorithm, and only using caching enables a 1.47--3.99x
speedup over the unoptimized algorithm.
Enabling caching decreases the number of bytes transmitted to the
key-value store by between 1.96--12.2x.
In what follows, when we refer to our \AMPC{} MIS algorithm we mean
the variant that uses both caching and multithreading unless mentioned
otherwise.

\begin{figure}[!t]
\begin{center}
\vspace{-0.2em}
\includegraphics[scale=0.55]{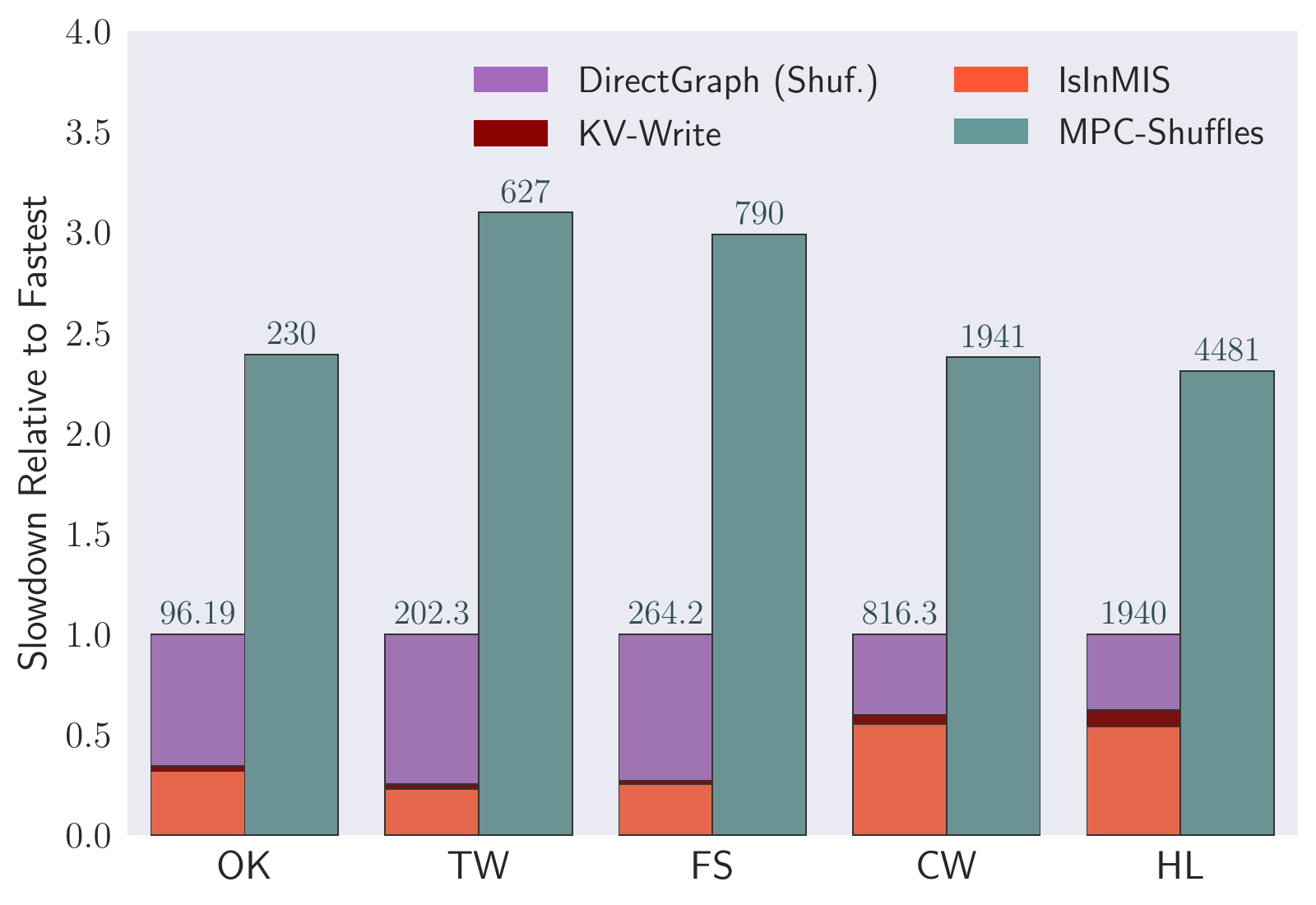}
\caption{\label{fig:mis_runtimes}
\revised{
Normalized running times for \AMPC{} and \MPC{} Maximal Independent
Set implementations. Each bar is annotated with the parallel running
time on top of the bar.
}
}
\end{center}
\vspace{-1em}
\end{figure}

\myparagraph{Running Time and Speedup}
Figure~\ref{fig:mis_runtimes} plots the normalized
running times of the \AMPC{} and \MPC{} MIS algorithms, and breaks
down the running time of the \AMPC{} algorithm into the three steps
indicated in Figure~\ref{fig:ampc_mis}. For smaller graphs, the time
spend shuffling to construct the directed graph is between 2.06--3.24x
more than the time spent in the \textsc{IsInMIS} search procedure. For
larger graphs, the cost of the search procedure is between 1.38--1.43x
more costly than building the directed graph due to a greater volume
of queries (see Section~\ref{sec:discussion}). Writing the directed graph
to the key-value store takes a small fraction of the time (at most
8\%).

We observed that the optimized \AMPC{} algorithm is always faster than
the \MPC{} algorithm, and achieves between 2.31--3.18x speedup over
the \MPC{} algorithm. The largest speedup is for the ClueWeb graph.
The slow performance of the \MPC{} algorithm on this graph is due to
skew in the join since the ClueWeb graph has many high degree vertices
with degrees larger than 10 million, and as large as 75.6 million. Our
speedups over the \MPC{} algorithm can be attributed to both
performing fewer shuffles, and shuffling fewer bytes of data overall
(see Figure~\ref{fig:mis_communication}). We also obtain this speedup
due to performing a modest amount of total communication to the
key-value store, which we discuss more in
Section~\ref{sec:discussion}.

}

%% file: inputs/fig_ampc_mis.tex
\definecolor{light-gray}{gray}{0.95}
\definecolor{dark-gray}{gray}{0.25}

\begin{figure*}[!ht]
\centering

\begin{minted}[fontsize=\small,linenos=true,numbers=left,autogobble,xleftmargin=0.2\textwidth,xrightmargin=0.2\textwidth,frame=single,escapeinside=||]{cpp}
// Uses hashing to determine a priority for each node.
uint64 NodePriority(NodeId node_id);

class IsInMIS : DoFn<KV<NodeId, Node>, NodeId> {
  IsInMIS(const string& kv_store_id) {
    kv_store = KVStore(kv_store_id); }
  void Do(const KV<NodeId, Node>& node,
          const EmitFn<NodeId>& emit) {
    if (InMIS(node.key, node.value)) { emit(node.key); }}
  // Check if any directed neighbors are in MIS. If
  // none are in, this node is in.
  bool InMIS(NodeId node_id, const Node& node) {
    for (NodeId neighbor_id : node.neighbors()) {
      // Fetch the neighbor's neighbors.
      Node neighbor = Lookup(neighbor_id);
      // Recursively query.
      if (InMIS(neighbor_id, neighbor)) return false;
    }
    return true;
  }
  // Synchronously query the key-value store for the
  // given ID, and return deserialized node.
  Node Lookup(NodeId id);
  KVStore* kv_store;
};

PCollection<NodeId> MIS(
	const PCollection<KV<NodeId, Node>>& graph) { |\label{line:graphinput}|
  // (1) Sort vertex neighborhoods based on priority.
  // Only preserve edges to higher-priority neighbors.
  PCollection<KV<NodeId, Node>> directed_graph =
      DirectEdgesUsingPriority(graph, NodePriority);
  // (2) Write directed graph to the key-value store.
  auto kv_store_id =
      WriteToKVStore(directed_graph);
  // (3) Apply the IsInMIS DoFn over nodes.
  PCollection<NodeId> ind_set = directed_graph.ParDo(
      IsInMIS(kv_store_id));
  return ind_set;
}
\end{minted}
\vspace{-0.05in}
\caption{\revised{Pseudocode for our \AMPC{} MIS algorithm.}}
\label{fig:ampc_mis}
\end{figure*}

%% file: inputs/fig_mpc_mis.tex
\begin{figure*}[!ht]
\centering
\begin{minted}[fontsize=\small,linenos=true,numbers=left,autogobble,xleftmargin=0.2\textwidth,xrightmargin=0.2\textwidth,frame=single,escapeinside=||]{cpp}
// Uses hashing to determine a priority for each node.
uint64 NodePriority(NodeId node_id);

PCollection<NodeId> MIS(
    PCollection<KV<NodeId, Node>> graph) {
  // Repeatedly extract a rootset
  vector<PCollection<NodeId>> independent_set;
  while (graph.numNodes() > 0) {
    // (1) Find all nodes that have priority lower
    // than their all neighbors. Since each node knows
    // its neighbors, and the priorities are computed
    // using hashing, this does not require a shuffle.
    PCollection<KV<<NodeId, Node>> new_set =
      LocalMinima(graph, NodePriority);
    independent_set.push_back(new_set.Keys());
    // (2) Compute node ids of the nodes in new_set
    // and their neighbors (no shuffle).
    PCollection<NodeId> to_remove =
        IdsOfNodesAndNeighbors(new_set);
    // (3) Mark which nodes should be removed. This
    // requires joining graph with node ids in
    // to_remove (1 shuffle).
    PCollection<KV<NodeId, pair<Node, bool>>>
	marked_graph =
	    MarkNodesToRemove(graph, to_remove);
    // (4) Each marked node x emits <x, y> and
    // <y, x> for each neighbor y. This computes all
    // edges to be deleted (no shuffle).
    PCollection<KV<NodeId, NodeId>> edges_to_delete =
	FindDeletedEdges(marked_graph);
    // (5) Update the graph by removing marked nodes
    // and their incident edges. Requires joining the
    // graph with edges_to_remove (1 shuffle).
    graph = RemoveNodesAndEdges(
	marked_graph, edges_to_remove);
  }
  // Flatten vector<PCollection<NodeId>> to
  // a PCollection<NodeId>
  return Flatten(independent_set);
}
\end{minted}
\vspace{-0.05in}
\caption{\revised{Pseudocode for the rootset-based \MPC{} MIS algorithm.}}
\label{fig:mpc_mis}
\end{figure*}

%% file: inputs/mm_experiments.tex
\subsection{Maximal Matching}\label{sec:mmexps}

\smparagraph{\AMPC{} Implementation and Optimizations}
Similarly to MIS, we implement algorithms finding the
lexicographically first matching for a random ordering of edges.  Our
\AMPC{} implementation is the constant round algorithm that we
designed in Section~\ref{sec:matching}, which we refer to as
\defn{AMPC-MaximalMatching}. The algorithm is implemented similarly to
the \AMPC{} MIS algorithm shown in Figure~\ref{fig:ampc_mis}.

\revised{
The main differences are that (i) the graph stored in the key-value
store does not direct the edges, but instead sorts the edges based on
random priorities assigned to each edge and (ii) instead of applying
an query process independently on each edge, we iteratively query
edges incident to each vertex $u$ in order of increasing priority.
Finally, instead of caching results \emph{per-edge}, we observed that
it suffices to maintain a value \emph{per-vertex}, since if an edge
incident to vertex $u$ with priority $P$ is not in the matching, this
indicates that all edges incident to $u$ with priority less than $P$
are also not in the matching.  Specifically, the cache stores for each
vertex a \textsc{NodeId}, and an enum indicating whether this id is
the highest priority neighbor that is finished, the matched neighbor,
or whether this vertex has not been searched yet.
}

\myparagraph{\MPC{} Implementation}
We implement a rootset-based maximal matching
algorithm, which is very similar to our MIS algorithm in the \MPC{}
setting.
Similarly to MIS, in each round, this algorithm adds to the matching
all edges whose priority is smaller than the priority of all its
adjacent edges and removes matched edges together with their
endpoints. We implement this algorithm in Flume C++ almost identically
to the MIS algorithm shown in Figure~\ref{fig:mpc_mis}. Once the graph
contains at most $s$ edges, where $s$ is a tunable parameter, it is
sent to a single machine, which finds the remaining edges of the
matching.  In our experiments, we experimentally determined that
setting $s=5\times 10^{7}$ gives a good tradeoff between the overhead
of spawning a new round of the algorithm, and the time taken to
process the graph on a single machine. Similarly to the MIS algorithm,
this algorithm takes $O(\log n)$ rounds with high probability.

\begin{figure}[!t]
\begin{center}
\vspace{-0.2em}
\includegraphics[scale=0.55]{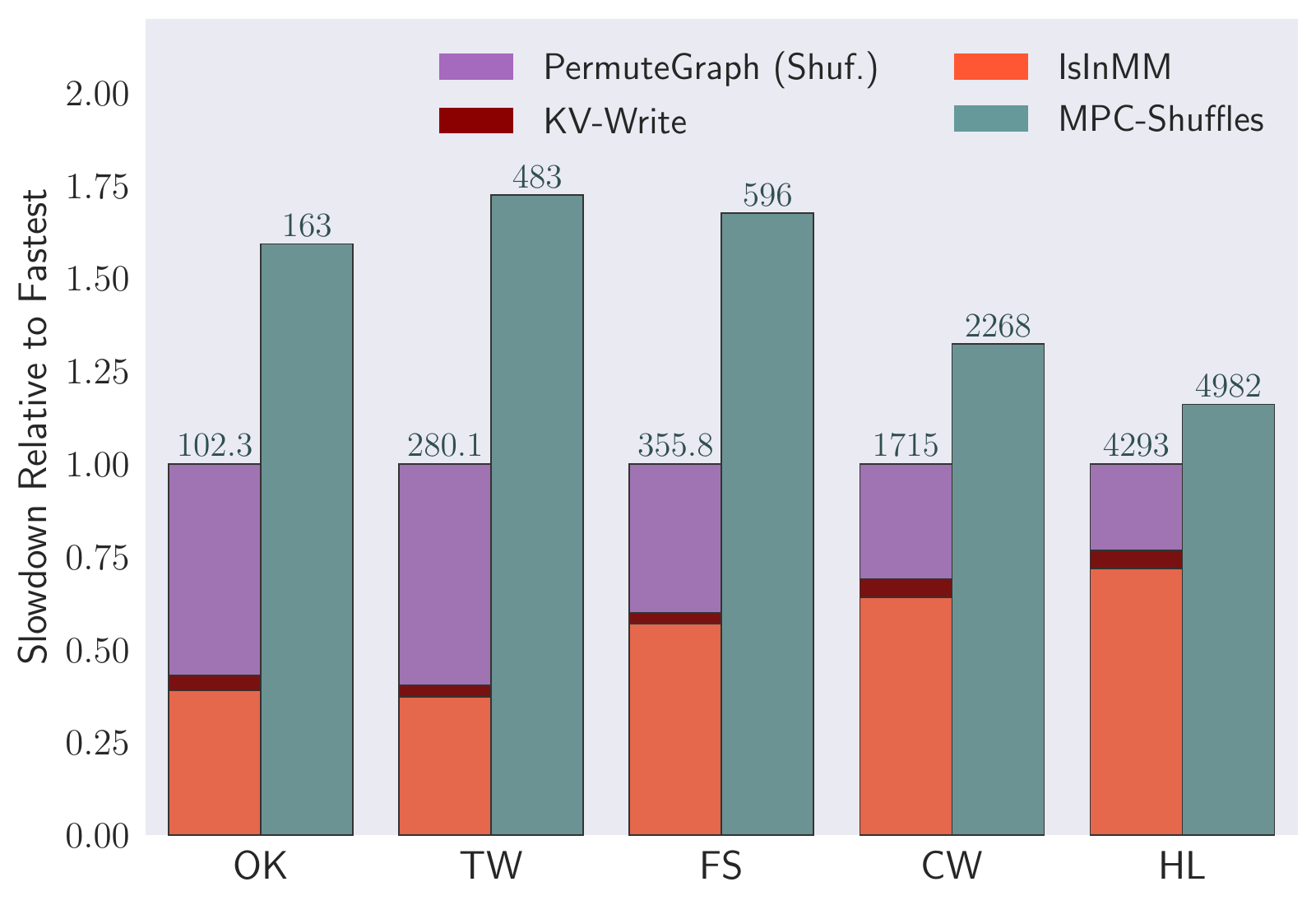}
\caption{\label{fig:mm_runtimes}
\revised{
Normalized running times for \AMPC{} and \MPC{} Maximal Matching
implementations. Each bar is annotated with the parallel running time
on top of the bar.
}
}
\end{center}
\end{figure}

\revised{
\myparagraph{Round-Complexity}
In Table~\ref{table:shuffles} we report the number of shuffles used by
our \AMPC{} MM algorithm and our \MPC{} MM algorithms. Like our
\AMPC{} MIS algorithm, our \AMPC{} MM algorithm only performs a single
shuffle to construct the edge-permuted graph. The remaining round is a
cheap map step over the original graph. Our \MPC{} MM algorithm
requires two shuffles per phase of the algorithm. We observe that the
algorithm performs a similar number of phases (and thus shuffles)
compared to the \MPC{} MIS algorithm.

\myparagraph{Optimizations}
Next, we considered the effect of optimizations on our \AMPC{} MM
algorithm. We found that multithreading was always beneficial, and did
not evaluate its effect in more detail. We were able to evaluate the
algorithm without using caching on OK, TW, and FS, but the algorithm
without caching took longer than 4 hours on CW and HL. For these
graphs, enabling caching reduced the total number of bytes read from
the key-value store by between 2.65--8.81x and improved the running
times by between 1.42--1.95x over the algorithm using just
multithreading.

\myparagraph{Running Time}
Figure~\ref{fig:mm_runtimes} reports the results of our comparison on
all of our real-world graph inputs, and reports a breakdown of the
different phases of our \AMPC{} MM algorithm. The breakdown is grouped
into the three steps of the algorithm---building the edge-sorted graph,
writing it to the key-value store, and running searches using the
\textsc{IsInMM} procedure---similar to the breakdown for our MIS
algorithm in Section~\ref{sec:miscasestudy}. We observe that copying
the graph takes somewhat longer than the MIS algorithm, which is due
to the fact that the all edges are present in this graph. The
breakdown otherwise exhibits a similar trend to the MIS ones.
Compared to the \MPC{} algorithms, the \AMPC{} algorithm is always
faster, and achieves between 1.16--1.72x speedup over the \MPC{}
algorithm. The main reason for the lower speedup compared to that of
the MIS algorithm is the larger cost of the search procedure, and the
larger amount of data shuffled to build the graph written to the
key-value store.
The cost of the \textsc{IsInMM} procedure increases with the graph
size, which is due to the increased number of queries made to the
key-value store (discussed more in Section~\ref{sec:discussion}).
}

%% file: inputs/msf_experiments.tex
\revised{
\subsection{Minimum Spanning Forest}\label{sec:msfexps}

\myparagraph{\AMPC{} Implementation}
Our \AMPC{} algorithm is based on the constant round algorithm that we
designed in Section~\ref{sec:msf}. We empirically found that
implementing a single search procedure on the graph without
ternarization is sufficient to shrink it to a very small size.
Conceptually our implementation has three main parts:

The first part sorts the edges incident to each vertex by their
weights. This graph is then written to the key-value store. The
algorithm then applies Prim's algorithm at each vertex $v$ using the
key-value store to fetch newly visited vertices. The search runs until
either the vertex hits a threshold in terms of the total amount of
edges it examines, or if it visits a vertex with higher priority. Each
search from a vertex $u$ emits a tuple containing itself and the
\textsc{NodeId} of every lower priority vertex it visits, which is
used to contract the graph in the next step. It also emits all of the
MSF edges that it observes during the search.

The second part groups the tuples from the previous part by the
visited vertex $u$, and combines them to select the visitor with the
\emph{highest priority} among all visitors to $u$.  The algorithm then
writes this map from \textsc{NodeId}s to visitor \textsc{NodeId}s to
the key-value store, and applies \emph{pointer jumping} to contract
the directed trees induced by the visited relationships. The final
output of this part is a contraction mapping sending \textsc{NodeId}s
to \textsc{NodeId}s (tree roots).

The final part contracts the graph based on the mapping obtained by
the previous step, and applies an in-memory MSF algorithm on the
contracted graph. The graph-contraction step is implemented using two
shuffles in Flume.

\myparagraph{\MPC{} Implementation}
We implement the classic \boruvka{}'s algorithm which finds a minimum
spanning forest in $O(\log n)$ rounds in the \MPC{} model. In each
phase of the algorithm, every vertex randomly colors itself either
\emph{red} or \emph{blue}. Each blue vertex computes the minimum
weight edge incident to it, and if this neighbor is red, then the
vertex contracts to the neighbor, otherwise the vertex does not
contract. The contraction routine is the same one used in our \AMPC{}
algorithm above. The algorithm iterates these phases until the number
of edges in the graph goes below $5\times 10^{7}$, at which point it
applies an in-memory MSF algorithm.

\begin{figure}[!t]
\begin{center}
\vspace{-0.2em}
\includegraphics[scale=0.55]{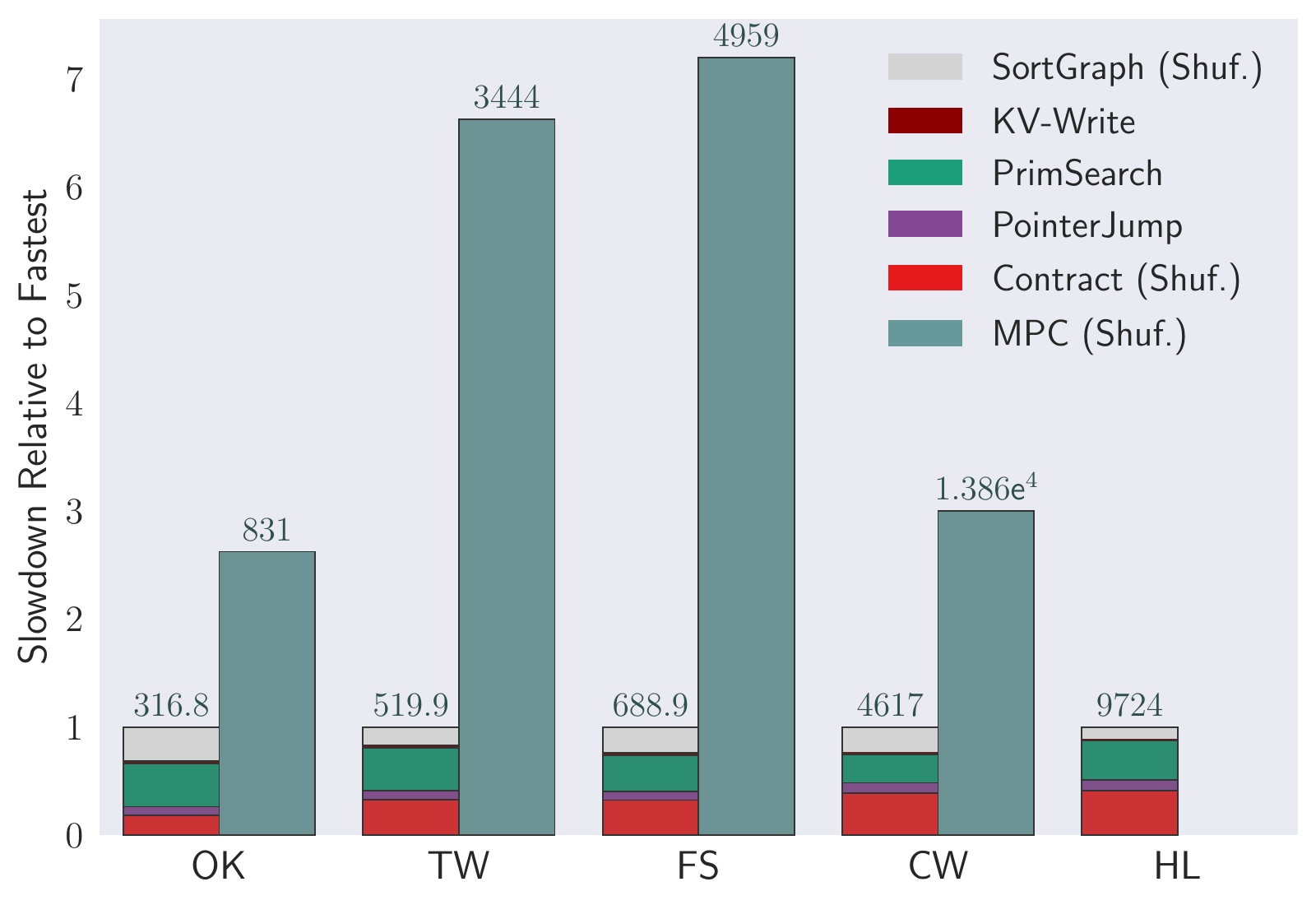}
\caption{\label{fig:msf_runtimes}
\revised{
Breakdown of running time for \AMPC{} and \MPC{} implementations of
Minimum Spanning Forest.
}
}
\end{center}
\end{figure}

\myparagraph{Round-Complexity}
Table~\ref{table:shuffles} reports the number of shuffles used by the
\AMPC{} and \MPC{} MSF algorithms. The \AMPC{} algorithm requires 5
shuffles, 1 for each of \{graph construction, combining on visited
vertices, pointer-jump construction\} and three shuffles to contract
the graph. The \MPC{} algorithm performs three shuffles per phase of
the algorithm (each phase contracts the graph); empirically,
\boruvka{}'s algorithm takes between 11--28 phases. The number of
phases is much higher than in the \MPC{} MIS or MM algorithms since
each phase of \boruvka{} only shrinks the number of vertices (not the
number of edges) in the graph by a constant factor in expectation, and
thus many phases are required.

\myparagraph{Running Time}
Figure~\ref{fig:msf_runtimes} shows the normalized running times for
the \AMPC{} and \MPC{} MSF algorithms. We note that we were unable
to obtain a result for the \MPC{} algorithm within 4 hours for the HL
graph. We breakdown the running time of our \AMPC{} algorithm into
five phases, namely the time to: sort the graph (SortGraph); perform
writes to the key-value store (KV-Write); perform Prim's searches
(PrimSearch); pointer-jump (PointerJump); and lastly the \MPC{}
routine used to contract the graph (Contract).

We observe that unlike in our MIS and MM algorithms, the largest
fraction of the time is spent on graph contraction (on the HL graph,
where the search time and contraction time are the closest, 15\% more
time is spent contracting). Performing pointer-jumping consistently
takes about 10\% of the overall time. Our implementation of
pointer-jumping simply repeatedly queries the parent of a vertex until
it hits a tree root.  Although the worst-case depth of this algorithm
could be as much as $O(n)$, in practice, the trees constructed by the
algorithm are very shallow (we observed a maximum query length of 33
over all graphs).

Compared to the \MPC{} algorithm, our \AMPC{} algorithm is always
faster, and achieves between 2.6--7.19x speedup. The running times for
both of our implementations of this problem are significantly slower
than for the other problems studied in this paper, primarily because
of the costly graph contraction procedure used in both algorithms.
We note that we tried to optimize this graph contraction step using
lookups to a key-value store, but were unable to obtain any
significant speedup over the \MPC{} routine.
}

%% file: inputs/cycle_experiments.tex
\subsection{1-vs-2-Cycle}\label{sec:twocycleexps}

\smparagraph{\AMPC{} and \MPC{} Implementation} The 1-vs-2-Cycle
problem is to distinguish whether an input graph consists of a single
cycle on $n$ vertices, or 2 cycles on $n/2$ vertices each.  This is a
canonical problem widely believed to be hard in the \MPC{} model,
i.e., the 1-vs-2-Cycle conjecture states that solving it requires
$\Omega(\log n)$ rounds in the \MPC{} model.  At the same time, the
problem admits a very simple algorithm that requires only $O(1)$
rounds in the \AMPC{} model~\cite{AMPC}. While the problem is of
purely theoretical interest, comparing \MPC{} and \AMPC{} algorithms
for the problem reveals the potential speedups from using the \AMPC{}
model.

The $O(1)$ round \AMPC{} algorithm for this problem is based on
sampling vertices with probability $O(n^{-\epsilon/2})$ and searching
outward from each vertex until another sampled vertex is hit.  Then,
the graph is contracted to a graph on the sampled vertices, with edges
between adjacent samples that discover each other.  We refer the
reader to Behnezhad~\etal{}~\cite{AMPC} for the full details of the
algorithm, and analysis.

We implemented this sampling-based algorithm in Flume C++ and refer to
it as \defn{\AMPC{}-1-vs-2-Cycle}, and evaluated its performance on a
family of $2\times k$-cycle graphs. Our implementation performs a
single round of the search procedure, sampling vertices with
probability $1/1024$, and solves the subsequent contracted graph on
a single machine.  We compare with a general-purpose
\MPC{} connectivity algorithm (\defn{CC-LocalContraction}) based on
performing local-contractions, which prior work found to be the
fastest \MPC{} connectivity implementation across a wide range of
graphs~\cite{cc-contractions}.

\revised{
\smparagraph{Results} On average, \AMPC{}-1-vs-2-Cycle achieves
between 3.40--9.87x speedup over CC-LocalContraction, with the
speedups increasing as the number of vertices increases. The main
savings (and improvement in the running time) comes from fewer
shuffles, and fewer bytes shuffled. The \AMPC{} algorithm requires a
single shuffle used to write the graph to the key-value store. The
\MPC{} algorithm reduces the length of the cycle by roughly a factor
of 2.59--3x in each iteration of the algorithm (2.69x on average).
Each iteration contracts the graph, which requires 3 shuffles. The
\MPC{} algorithm uses 4--9 iterations across all cycle inputs (12--27
shuffles).
}

%% file: inputs/discussion.tex
\revised{
\subsection{Discussion}\label{sec:discussion}
We end this section by discussing possible bottlenecks in our
implementations, and directions for future work.

\begin{figure}[!t]
\begin{center}
\vspace{-0.2em}
\includegraphics[scale=0.55]{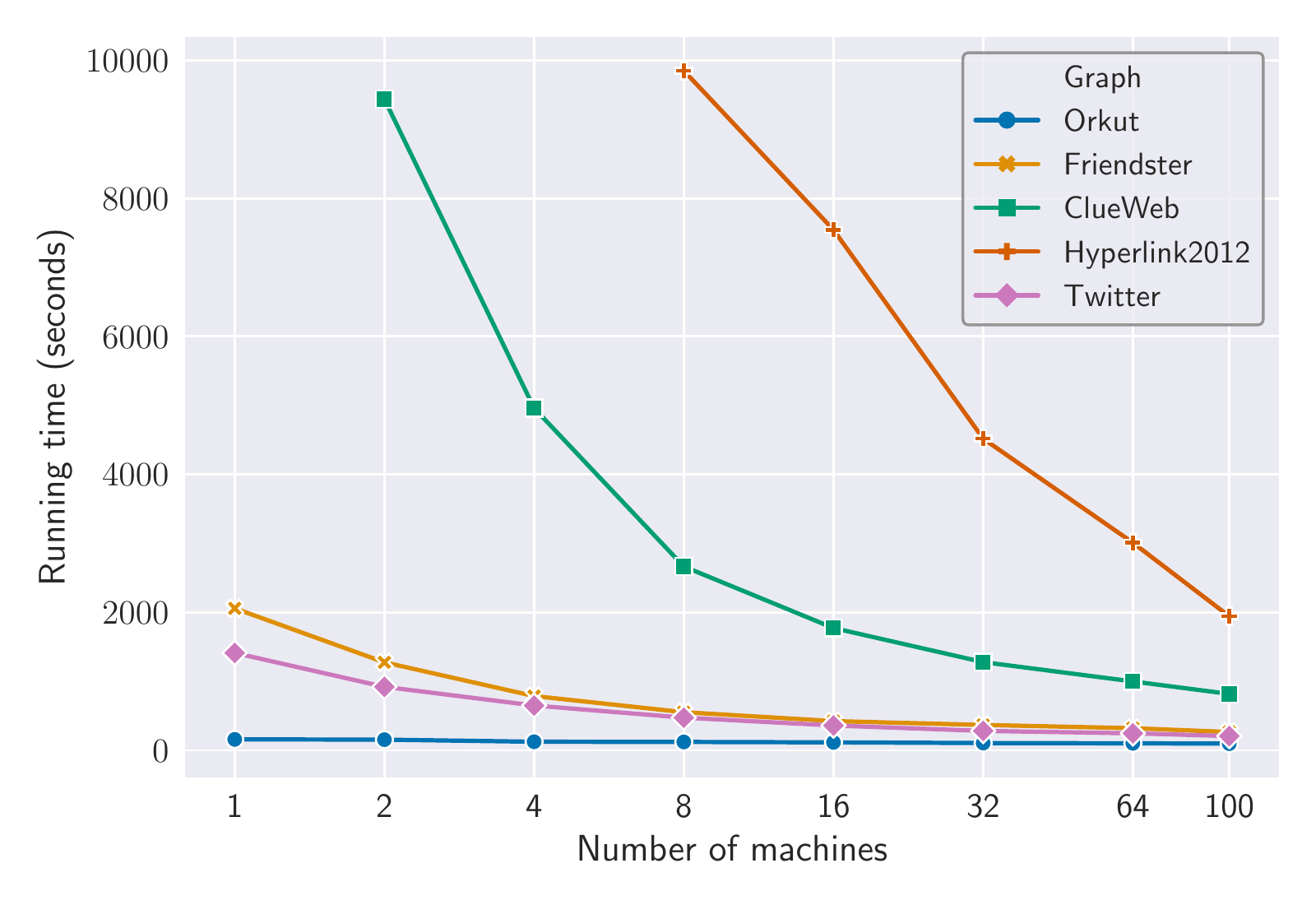}
\vspace{-0.6em}
\caption{\label{fig:mis_scaling}
\revised{
Self-speedup of the \AMPC{} MIS algorithm when run on between 1--100
machines.
}
}
\end{center}
\vspace{-1em}
\end{figure}
\myparagraph{Scaling}
We evaluated the scaling (or self-speedup) or our \AMPC{} MIS
algorithm to check that our \AMPC{}
algorithms obtain speedups when varying the number of machines. We
were unable to obtain 1-machine times for the ClueWeb graph and 1, 2,
or 4-machine times for the Hyperlink2012 graph within 4 hours. For the
smaller graphs, the 100-machine time is between 1.64--7.76x faster
than the 1-machine time. The speedups are better for larger graphs,
since there is more work to do relative to the overhead of spawning
rounds and shuffles. For ClueWeb, the 100-machine time
obtains a 11.5x speedup over the 2-machine time, and for Hyperlink2012
the 100-machine obtains a 5x speedup over the 8-machine time. We
believe that one of the reasons we do not obtain linear speedup may be
due to saturating the network bandwidth when querying the key-value
store, which we discuss next.

\begin{figure}[!t]
\begin{center}
\vspace{-0.2em}
\includegraphics[scale=0.55]{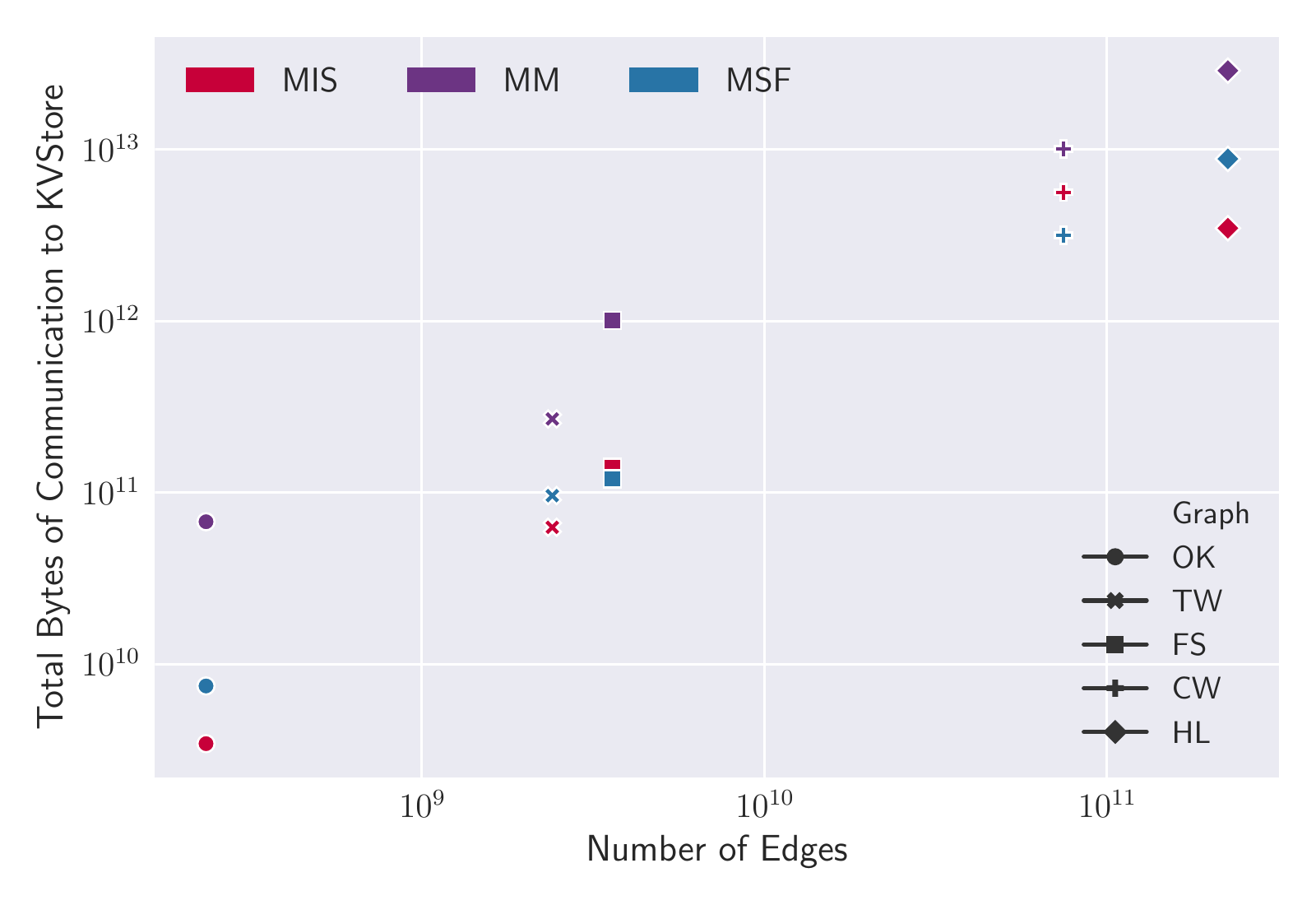}
\vspace{-0.6em}
\caption{\label{fig:bytes_communication}
\revised{
Total bytes of communication to the key-value store by \AMPC{}
algorithms.
}
}
\end{center}
\vspace{-1em}
\end{figure}
\myparagraph{Key-Value Store Communication}
Figure~\ref{fig:bytes_communication} plots the the number of edges in
our real-world inputs on the x-axis and the total bytes of data
communicated to the key-value store on the y-axis. We observe that for
all of the problems there is a consistent linear trend in terms of the
total amount of communication with respect to the
number of edges. Although the total number of bytes communicated is
one upper bound on the total amount of memory used by our algorithms,
we point out that in practice, the total memory usage is closer to the
graph size (which is at most 2TB for the Hyperlink2012 graph).
Finally, we observed that in all cases, the maximum throughput we
obtain on our network in our experiments is about
80Gb/sec (which amounts to just under 1 Gb/sec per machine). It seems
likely that the time spent on the search phase of each algorithm could
be accelerated by exploiting high throughput networks, and this is an
interesting direction to investigate in future work.

\setlength{\tabcolsep}{3pt}
\begin{table}[!t]
\small
  \centering
\begin{tabular}{l|c|c|c|c|c|c|c|c}
\toprule
{\bf Algorithm}                    & {\bf OK} & {\bf TW} & {\bf FS} & {\bf CW} & {\bf HL} & {\bf $2 e^{8}$} & {\bf $2e^{9}$} &{\bf $2 e^{10}$} \\
        \midrule
{2-Cyc. (RDMA) }  & -- & -- & -- & -- & -- & 1    & 1    & 1\\
{2-Cyc. (TCP/IP)} & -- & -- & -- & -- & -- & 1.74 & 3.75 & 5.90 \\
{\MPC{} 2-Cyc.}   & -- & -- & -- & -- & -- & 3.40 & 6.70 & 9.87 \\
\midrule
{MIS (RDMA) }     & 1    & 1    & 1    & 1    & 1    & -- & -- & -- \\
{MIS (TCP/IP) }   & 1.85 & 1.63 & 1.50 & 1.68 & 1.71 & -- & -- & -- \\
{\MPC{} MIS }     & 2.39 & 3.04 & 2.98 & 2.37 & 2.30 & -- & -- & -- \\
\end{tabular}
\caption{\shep{Normalized running times of 1-vs-2-Cycle  (2-Cyc.) and MIS
algorithms. The \AMPC{} algorithms are the first two rows in each
group and communicate with the key-value store using either RDMA or
TCP/IP. The \MPC{} implementation is the last row in each group.
$2e^{k}$ refer to the $2 \times 10^{k}$ graphs.}}
\label{table:rdmavstcpvsmpc}
\end{table}

\shep{
How much of the speedup over \MPC{} algorithms in our experiments is
due to using RDMA-based communication? To study this question we
replaced all of the key-value store communication within the \AMPC{}
algorithms with RPCs sent over TCP/IP.
Table~\ref{table:rdmavstcpvsmpc} shows the results of this evaluation
for the 1-vs-2-Cycle and MIS problems. We see that the impact of using
RDMA is more significant for 1-vs-2-cycle, since the lower latency of
RDMA enables searches around the cycle to terminate more quickly. The
slowdown for the MIS problem is more modest, with the TCP/IP based
algorithm only being 1.67x slower than the RDMA-based algorithm on
average.
These results suggest that although RDMA-based communication is
important in achieving good running times, it can safely be replaced by RPCs
sent over TCP/IP, and the resulting algorithms will still outperform
fast \MPC{} baselines.
}

\shep{
\myparagraph{\AMPC{} vs \MPC{} Algorithms}
One of the main messages of our paper is that \AMPC{} algorithms can
deliver better performance than state-of-the-art \MPC{} algorithms,
without sacrificing fault-tolerance properties. There are a number of
reasons for this performance improvement.  First, \AMPC{} algorithms
require significantly fewer rounds compared to fast \MPC{} baselines
as shown in Table~\ref{table:shuffles}.  Second, \AMPC{} algorithms
also require shuffling significantly less data per round, which
results in fewer disk writes (e.g., see
Figure~\ref{fig:mis_communication}). Finally, using RDMA as a
communication mechanism results in excellent performance as the
algorithm can effectively replace the shuffle-based communication
mechanism which uses costly disk writes with a fast, low-latency
protocol. Ultimately, we see the \AMPC{} approach as an interesting
middle-ground between systems that communicate through persistent
storage (like MapReduce, Hadoop and Flume-C++) and are thus robust to
preemptions, and systems that run fully in memory, which deliver
better performance at the cost of not tolerating preemptions well and
requiring high priority (reserved) resources.
}

\myparagraph{Applicability}
We believe the \AMPC{} model is a promising platform for problems
which can be expressed using graph exploration, such as connected
components, sub-structure problems, and problems involving random
walks.

\emph{Connected Components}
One specific challenge stemming from this paper is to obtain a result
for connectivity that achieves speedups over the state-of-the-art
\MPC{} algorithm~\cite{cc-contractions}. We tried to apply our MSF
algorithm over a graph with random edge weights, but were not able to
obtain significant speedups over this \MPC{} result due to the high
cost of graph contraction on the first step (contracting the initial
graph takes about 2/3 of the overall running time).

\emph{Sub-structure Extraction.} The $k$-core, $k$-truss, and related
problems have garnered a large amount of interest due to their
importance in community detection~\cite{batagelj03cores,
sariyuce2016incremental, Sariyuce2013, sariyuce2018local,
DBLP:journals/tkde/WenQZLY19, DBLP:conf/sigmod/0001Y19}. It would be
interesting to study whether we can solve these problems $O(1)$ rounds
in the \AMPC{} model.

\emph{Random-walk and Embedding.}
The \AMPC{} model can potentially help accelerate random-walk based
problems, such as PageRank and Personalized PageRank~\cite{
bahmani2010fast, gleich2015pagerank,page1999pagerank} since it
efficiently supports random access. Graph embeddings are another
impactful data-mining area where expressing random-walk algorithms,
such as DeepWalk~\cite{perozzi14deepwalk}, LINE~\cite{tang2015line},
and NetSMF~\cite{qiu2019netsmf} in \AMPC{} could allow scaling to
massive graphs in a reliable, fault-tolerant way.
}

%% file: inputs/relwork.tex
\section{Related Work}\label{sec:related}

\myparagraph{Massively Parallel Algorithms}
There has been a huge body of work on developing efficient, low
round-complexity distributed graph algorithms in the Massively
Parallel Computation model~\cite{ DBLP:conf/spaa/AhnG15, DBLP:conf/stoc/AndoniNOY14, andoniparallel, DBLP:conf/soda/AssadiBBMS19, DBLP:conf/nips/BateniBDHKLM17, DBLP:conf/spaa/BehnezhadDETY17, DBLP:journals/corr/abs-1802-10297, DBLP:conf/focs/BehnezhadHH19, DBLP:conf/stoc/CzumajLMMOS18, DBLP:conf/podc/GhaffariGKMR18,  koutris2018algorithmic, DBLP:conf/spaa/LattanziMSV11, DBLP:conf/spaa/RoughgardenVW16}, including many other papers. In this paper, our focus
is on recently proposed Adaptive Massively Parallel Computation
model~\cite{AMPC}, which was also recently studied from a lower-bound
perspective by Charikar~\etal{}~\cite{lowerboundAMPC}.

\myparagraph{Distributed Graph Processing}
Motivated by the need to process very large graphs, there have been
many distributed graph processing frameworks developed in the
literature (e.g.,
\cite{gonzalez2012powergraph,low2010graphlab,pregel}
among many others). We refer the reader to~\cite{McCune2015,Yan2017}
for surveys of existing work in this area.

%% file: inputs/conclusion.tex
\section{Conclusion and Open Problems}\label{sec:conclusion}

In this paper we presented new graph algorithms with constant
round-complexity in the Adaptive Massively Parallel Computation model.
Our theoretical results for these problems address three open problems
posed in~\cite{AMPC}, essentially settling the round complexity of connectivity, minimum spanning forest and matching in this model.
\shep{Our practical evaluation of these algorithms shows that \AMPC{}
algorithms can deliver better performance than corresponding \MPC{}
algorithms, without sacrificing fault-tolerance properties.}

For future research, it would be interesting to understand whether
the total query complexity of these problems could be reduced
to linear even for sparse graphs, when $m = O(n)$.
It would also be interesting to understand whether the model can be used to solve other fundamental large-scale
problems faster, such as string matching, and whether it can be
applied to distance problems on graphs.

%% file: inputs/proofs.tex
\section{Deferred Proofs from Section~\ref{sec:msf}}\label{apx:msf}

\msfshrink*
\begin{proof}
Observe that the set of vertices that are not contracted are exactly
those vertices that do not encounter a lower-rank vertex in their
search, and do not fully search their connected component. The first
condition is because each vertex that encounters a lower-rank vertex
adds a directed edge to $F$ and is contracted to the root of its
directed tree in $F$. The second condition is since vertices that
fully explore their connected component are isolated, and are removed
from the contracted graph by Step~\ref{prim:contract}.

Let $T$ be the MSF of the input graph to
Algorithm~\ref{alg:truncatedprim} which is ternarized by assumption,
and let $G'(V',E')$ be the contracted graph output by this algorithm.
The Prim search for each vertex can be viewed as a connected
exploration in $T$. Each vertex that survives in $G'$ is a vertex that
stops its exploration due to case (1), namely, that it explores
$n^{\epsilon/2}$ vertices without finishing searching its component,
and without seeing a vertex that appears before it in $\pi$. There can
be at most $O(n^{1-\epsilon/2})$ such vertices, since each vertex that
survives in $G'$ uniquely acquires $n^{\epsilon/2}$ vertices from one
direction in $T$. A vertex can be acquired in at most three
directions, since $T$ is ternarized. Therefore, the total number of
surviving vertices in $G'$ is at most $3n/n^{\epsilon/2} =
O(n^{1-\epsilon/2})$ vertices, which is a factor of
$O(n^{\epsilon/2})$ fewer than $G$ as desired.
\end{proof}

Next, we show that the total number of queries made by the algorithm
is concentrated around its mean. One can show that
the expected number of queries per vertex is $O(\log n)$, as shown in
Lemma {8.2} of ~\cite{AMPC}, which considered a similar randomized
process for cycle connectivity. In the cycle connectivity algorithm,
the local search for each vertex simply walks along the cycle until
either the search grows too large, or a higher priority vertex in
$\pi$ is hit. To obtain a high probability bound, the authors of
~\cite{AMPC} showed that the cost of exploring the cycle was
equivalent to the cost of randomized quicksort, and then used the fact
that randomized quicksort runs in $O(n\log n)$ operations w.h.p.

Unfortunately, it seems difficult to map the Prim searches done by
each vertex to a cycle, and there are simple counterexamples showing
that the Prim searches cannot be mapped onto the Euler tour of the
MST. Instead, we relate the cost of the Prim searches to the analysis
of \emph{treaps}~\cite{DBLP:conf/focs/AragonS89}. The idea is to think
about query process as first building a treap on the ternarized MST
which we refer to as a \emph{\threep{}}, since vertices in the
\threep{} have degree $\leq 3$. Like a regular treap, the node with
highest priority is at the root of the \threep{} (it is a min-heap
with respect to rank in $\pi$), and each of its children contain the
node with highest priority in that subtree.  A simple proof by
induction shows that there is a unique \threep{} corresponding to a
given permutation $\pi$ and ternary tree $T$.

Intuitively, the same properties that hold for treaps should hold for
\threep{}s, since a node in a treap can split the subproblem
containing it two disjoint pieces, while a node in a \threep{} can
split the subproblem containing it into three disjoint pieces.  We
first show that the height of a \threep{} is $O(\log n)$ w.h.p., and
then relate the cost of a Prim search from a vertex $v$ to the size of
its subtree in the \threep{} defined by $T$ and $\pi$.

\begin{restatable}{lemma}{threepheight}\label{lem:threepheight}
Given a tree with $\Delta \leq 3$, $T$, the \threep{} defined by $T$
and $\pi$ where $\pi$ is a uniformly random permutation from $[n]
\rightarrow [n]$ has height $O(\log n)$ w.h.p.
\end{restatable}

\begin{proof}
The depth of a node $v$ is the number of nodes
from $v$ to the root of the \threep{}. Assume $T$ is connected (note
that adding edges connecting disjoint components only increases the
height of the resulting \threep{}).  A simple high probability bound
on the depth of an ordinary treap can be obtained as follows. Define
the indicator random variable $X_{i}^{j}$ which is $1$ if node $j$ is
an ancestor of $i$ and $0$ otherwise. The important fact is that for a
given $i$, the variables $X_i^{j}$ are independent. Let $X_i =
\sum_{j=1}^{n} X_{i}^{j}$. The bound then follows by a Chernoff bound
on $X_i = \sum_{j=1}^{n} X_{i}^{j}$, which is $O(\log n)$ in
expectation.

To modify this analysis to \threep{}s, we observe that we can define
variables $X_i^{j}$ similarly, and observe that these variables are
also independent. To see why the $X_i^{j}$ are independent, consider
two such random variables for a given node $i$, $X_i^{j}$ and
$X_i^{k}$. Now, consider the unique path in $T$ between $i\rightarrow
j$.  If $k$ falls on this path, then whether $j$ is an ancestor of $i$
only depends on whether $j$ has the highest probability on this path,
and not on whether $k$ has the highest probability on its path to $i$,
or on relative orderings of nodes within the path. Therefore, the
variables $X_i^{j}$ are independent. Applying a Chernoff bound as in
the case of an ordinary treap completes the proof.

For an overview of the high probability analysis of treaps, see for
example~\cite{dubhashi2009concentration}.
\end{proof}

\begin{restatable}{lemma}{queryandthreep}\label{lem:queryandthreep}
Given a weighted graph $G$ with $\Delta(G) \leq 3$, and a unique MSF
of $G$, $T$, let $R$ be a \threep{} generated by $T$ and a uniformly
random permutation over $[n]$, $\pi$. Let $R_v$ be the subtree rooted
at $v$ in $R$. Then, the number of queries made by the truncated Prim
search from a vertex $v$ in $G$ with respect to priorities in $\pi$ is
upper-bounded by $O(|R_v|)$, the size of $v$'s subtree in the
\threep{}, $R$.
\end{restatable}
\begin{proof}
Consider the $v$'s node in $R$, and consider running the truncated
Prim search in $G$ starting at $v$. This search is a deterministic
sequence of vertex visits starting at $v$ that visits a set of
vertices emanating from $v$ in $T$ until one of the three stopping
conditions in Step~\ref{prim:conditions} is met. For simplicity, we
can ignore condition (1), since ignoring it only increases the query
cost.  To summarize, consider running Prim's algorithm starting at $v$
until either the entire component containing $v$ is visited, or $v$
hits a vertex $u$ s.t. $\pi(u) < \pi(v)$.

To prove the claim, assume for the sake of contradiction that $v$'s
truncated Prim search includes a vertex outside $R_v$ in $v$'s
cluster. Note that this is the only way that the search can make more
than $O(R_v)$ queries, since the vertices in $T$ are ternarized.
However, this means that there is an edge from some $v' \in R_v$ that
connects to a vertex $w \in T$ s.t. $\pi(w) > \pi(v)$, and $w$ is not
in $R_v$. Since the $(v,w)$ edge forms a cycle in $T$ which is acyclic
by assumption, we have derived a contradiction. Therefore, the Prim
search started at $v$ only includes vertices within $R_v$, until
either $v$ traverses its parent edge, or traverses an edge $(w,a)$
where $w \in R_v$ and $a$ is an ancestor of $v$ in $R$. In both cases,
the Prim search visits a vertex with higher priority in $\pi$ and is
immediately terminated. Therefore, the total query cost of the
truncated Prim search starting at $v$ with respect to $\pi$ is
$O(|R_v|)$.
\end{proof}

Given the relationship between the query cost of the truncated Prim
search and the subtree size in a \threep{}, we are now ready to prove
the total number of queries made by the algorithm.

\msfqueries*
\begin{proof}
  Let $R$ be a \threep{} generated by $T$ and the uniformly random
  permutation chosen by Algorithm~\ref{alg:truncatedprim}, $\pi$.  By
  Lemma~\ref{lem:queryandthreep}, we can upper bound the total query
  cost by summing the size of each vertex $v$'s subtree, $R_v$ in $R$.
  By Lemma~\ref{lem:threepheight}, we have that the height of a
  \threep{} over $n$ vertices is $O(\log n)$ w.h.p. Therefore, each
  vertex participates in at most $O(\log n)$ subtrees w.h.p. The total
  query complexity can therefore be bounded as:
  \begin{equation*}
    \sum_{v \in V} |R_v| = \sum_{v \in V} \sum_{u \in R_v} 1 \leq
    \sum_{v \in V} O(\log n) = O(n \log n)
  \end{equation*}
  Note that Lemma~\ref{lem:queryandthreep} ignored stopping condition
  (1), which truncates the Prim search if it grows too large. However,
  Algorithm~\ref{alg:truncatedprim}, which respects the stopping
  condition, can only make fewer queries. Therefore,
  Algorithm~\ref{alg:truncatedprim} makes $O(n \log n)$ queries in
  total w.h.p.
\end{proof}

\primspacerounds*
\begin{proof}
  By Lemma~\ref{lem:msfqueries}, the total number of queries made by
  Algorithm~\ref{alg:truncatedprim} is $O(n\log n)$ w.h.p., which
  bounds the total space used by the algorithm. We must now argue that
  the algorithm can be implemented so that each machine makes at most
  $O(n^{\epsilon})$ queries. This can be done by the same argument
  made to bound the number of queries made per machine in the forest
  connectivity algorithm in~\cite{AMPC}.

  Specifically, Lemma 8.4 of~\cite{AMPC} shows that if one throws $n$
  weighted balls, where the maximum weight of a ball is
  $O(n^{\epsilon/2})$, and the average weight of a ball is $O(\log
  n)$, the load of a machine which randomly selects $n^{\epsilon}$
  balls without replacement is $O(n^{\epsilon})$ w.h.p. Both of these
  conditions are satisfied in Algorithm~\ref{alg:truncatedprim}. The
  maximum number of queries made is $O(n^{\epsilon/2})$ since the
  search is truncated by stopping condition (1) if the cluster size
  exceeds $O(n^{\epsilon / 2})$, and the average size of a search is
  $O(\log n)$ by using Lemma~\ref{lem:msfqueries} and averaging over
  all vertices.

  The round complexity follows from the fact each step of the
  algorithm can be implemented in $O(1/\epsilon)$ rounds. Generating a
  random permutation can be done in $1$ round by using a random source
  within each machine to generate a random number from a suitably
  large range per vertex (e.g., $[n^2]$), and using the order based on
  these numbers as the random permutation. Running Prim's algorithm
  locally per vertex is done in a single round. Finally, contracting
  the graph can be reduced to sorting and removing duplicates, both of
  which can be implemented in $O(1/\epsilon)$ rounds of \MPC{}, and
  therefore the same number of rounds of \AMPC{}.
\end{proof}

\msfconstant*
\begin{proof}
  Suppose the input graph is dense, i.e. $m = O(n^{1+\epsilon})$. In
  this case, the algorithm simply runs the algorithm from
  Proposition~\ref{prop:msfspaa}, which runs in
  $O((1/\epsilon)\log\log_{(m+n)/n} n) =
  O((1/\epsilon)\log(1/\epsilon))$ rounds of $\AMPC{}$ and $O(m+n)$
  space.

  If the graph is sparse, the algorithm first ternarizes the graph
  which can easily be done in $O(1/\epsilon)$ rounds by sorting. The
  ternarized graph has $O(m)$ vertices and $O(m)$ edges. The algorithm
  then calls Algorithm~\ref{alg:truncatedprim} on the ternarized
  graph, which runs in $O(1/\epsilon)$ rounds w.h.p. and $O(m\log n)$
  total space w.h.p. by Lemma~\ref{lem:primspacerounds}. The
  contracted graph output by this algorithm has $O(m)$ edges and
  $O(m^{1-\epsilon/2})$ vertices. Finally, the algorithm calls the
  algorithm from Proposition~\ref{prop:msfspaa} the contracted graph
  which runs in
  $O((1/\epsilon)\log\log_{m/m^{1-\epsilon/2}} m) =
  O((1/\epsilon)\log(1/\epsilon))$ as in the dense case. Therefore, in
  both cases, the algorithm runs in $O((1/\epsilon)\log(1/\epsilon))$
  rounds w.h.p. and $O(m \log n)$ total space w.h.p.
\end{proof}

%% file: inputs/flight.tex
\section{Computing F-light edges}\label{app:flight}
In this section we describe how to implement line~\ref{l:fl} of Algorithm~\ref{alg:msfreduction} in the AMPC model.
Before proceeding with the description, we review three basic tools that we are going to use.

\myparagraph{Lowest common ancestor} Let $T$ be a tree. If $u, w$ are two vertices of $T$, we use $T[u, w]$ to denote the unique path from $u$ to $w$ in $T$.
If $T$ is rooted in a vertex $r$, we define the \emph{level} of each vertex $v$ as the length (number of edges) of $T[v, r]$.
The \emph{lowest common ancestor} of $u, w \in T$, which we denote by $LCA(u, w)$ is the common vertex of $T[u, r]$ and $T[w, r]$ that has the highest level.
It is well-known that it is defined uniquely, and that $T[u, w]$ is obtained by concatenating $T[u, LCA(u, w)]$ and $T[LCA(u, w), w]$.

\myparagraph{Heavy-light decomposition} Let $T$ be a tree rooted in vertex $r$.
For each non-leaf vertex $v$ of $T$, we compute the sizes of the subtrees rooted at children of $v$, choose the subtree of the largest size (breaking ties arbitrarily) and mark the edge from $v$ to the child with the largest subtree as \emph{heavy}.
All remaining edges are \emph{light}.

It is easy to see that the heavy edges form disjoint paths, which we call \emph{heavy paths}.
The basic property of a heavy-light decomposition is as follows.
For each vertex $v \in T$, the path $T[v, r]$ consists of $O(\log n)$ light edges and $O(\log n)$ contiguous segments, each being a subpath of a heavy path.

For a vertex $u \in T$ and consider $T[u, r]$.
We say that a vertex $x$ on $T[u, r]$ is a \emph{pivot of $x$} if any of the following holds:
\begin{itemize}
\item $x$ is $u$ or $r$, or
\item there exists a light edge of $T[u, r]$, which is incident to $x$.
\end{itemize}

In our algorithm we will use the following easy lemma.

\begin{lemma}\label{lem:pivot}
Each vertex $u \in V(T)$ has $O(\log n)$ pivot vertices.
Let $a$ be an ancestor of $u$, and let $p$ be the lowest level pivot vertex on $T[u, a]$.
.
Then, the path $T[u, a]$ is a concatenation of $T[u, p]$ and $T[p, a]$, where $T[p, a]$ is fully contained within a heavy path.
\end{lemma}

\begin{proof}
Clearly, the number of pivot vertices distinct from $x$ and $r$ is at most twice the number of light edges, which immediately gives us a bound of $O(\log n)$.

To prove the second claim, let us first note that $p$ is uniquely defined.
Indeed, this follows from the fact that the levels of vertices on the path are distinct and $u$ is a pivot vertex.
If any edge on $T[p, a]$ was light, we would immediately get a contradiction with the choice of $p$, so the lemma follows.
\end{proof}

\myparagraph{Range-minimum queries}
Given an array $a_1, \ldots, a_k$, a \emph{range-minimum query} (RMQ) data structure is a data structure, which given two indices $1 \leq i \leq j \leq k$ computes the minimum among $a_i, \ldots, a_j$ in $O(1)$ time.
A possible approach is to compute an auxiliary array $b_{x, y}$ for $1 \leq x \leq k$ and $0 \leq y \leq \log_2 k$, where $b_{x, y} = \arg \min_{i = x, \ldots, \min(x+2^y-1, k)} a_i$.
Then, finding the minimum among $a_i, \ldots, a_j$ could be done by computing $t = \lfloor \log_2 (j-i+1) \rfloor$ and simply taking the minimum of $b_{i, t}$ and $b_{j-2^t+1, t}$.
We call the array $b$ an RMQ data structure for $a$.
Andoni et al.~\cite{rmqref} showed how to compute the RMQ data structure in the MPC model in $O(1)$ rounds using $O(k \log k)$ total communication.

We now describe our algorithm for finding $F$-light edges.
There are two ways in which an edge $uw \in E(G)$ can be qualified as $F$-light (see Definition~\ref{def:flight}).
First, $u$ and $w$ may belong to different connected components of $F$.
These edges can be detected easily by first finding connected components of $F$, that is a mapping from each vertex of $F$ to a unique identifier of its connected components.
Whether two vertices are in the same connected component can be then determined using just two queries.

In the following let us focus on the remaining case, when $u$ and $w$ belong to the same connected component of $F$, which we denote by $T$.
To determine whether $uw$ is $F$-light it suffices to compute the largest edge weight on $T[u, w]$.
To that end, we root $T$ arbitrarily.
Let us denote the root by $r$.
The result of rooting the tree is an array $p$ which maps each vertex other than $r$ to its parent in the rooted tree.

To find the maximum edge weight on $T[u, w]$, we first compute $LCA(u, w)$ and then reduce the problem to finding maximum weight edges on $T[u, LCA(u, w)]$ and $T[w, LCA(u, w)]$.
In the following consider $T[u, LCA(u, w)]$.

To find the maximum weight edge on $T[u, LCA(u, w)]$ we use Lemma~\ref{lem:pivot}.
Note that $LCA(u, w)$ is an ancestor of $u$.
Let $p$ be the lowest pivot vertex on $T[u, LCA(u, w)]$.
We are going to find maximum weight edges on $T[u, p]$ and $T[p, LCA(u, w)]$ separately.

In order to efficiently find the maximum within $T[u, p]$ we are going to precompute the maximum weight edge on $T[u, p]$ for each of the $O(\log n)$ pivot vertices $p$ of $u$.
On the other hand, $T[p, LCA(u, w)]$ is fully contained in a heavy path.
Hence, by precomputing a RMQ data structure for each heavy path, we can find the maximum weight edge by reading a constant number of values from this data structure.
The final algorithm is given as Algorithm~\ref{alg:flight}.
To complete this section, we prove the correctness of the algorithm and discuss its implementation in the AMPC model.

\begin{tboxalg}{$\mathrm{FindLightEdges}(G, F)$}\label{alg:flight}
\begin{algorithmic}[1]
\State Find connected components in $F$.\label{l:cc}
\State Root each connected component of $F$.
\State For each vertex in $F$, compute its level in the tree it belongs to.
\State Compute an Euler tour traversal of each tree $T$ of $F$.
\State Within the traversal sequence, assign to each vertex the weight equal to its level and compute an RMQ data structure for each sequence.
\State For each $uw \in E(G)$ such that $u$ and $w$ are in the same connected component of $F$, compute $LCA(u, w)$.
\State Compute the heavy-light decomposition of each tree $T$ in $F$ (mark each edge as heavy or light).
\State Compute the connected components of the heavy paths and an RMQ data structure for each heavy path.
\State For each vertex $v \in T \in F$, find the $O(\log n)$ heavy paths and light edges on the path from $v$ to the root of $T$.
\State For each $uw \in E(G)$, such that $u$ and $w$ are in the same connected component of $F$, compute the maximum edge weight on $F[u, LCA(u, w)]$ and $F[w, LCA(u, w)]$.
\State Use the weights computed in the previous step and the connected component identifiers computed in line~\ref{l:cc} to find $F$-light edges.
\end{algorithmic}
\end{tboxalg}

\begin{lemma}
Let $G = (V, E, w)$ be a weighted graph and $F \subseteq G$ be a tree.
Algorithm~\ref{alg:flight} correctly identifies $F$-light edges of $G$.
It can be implemented in $O(1)$ AMPC rounds using $O(n \log n)$ queries in total, where $n = |V|$.
\end{lemma}

\begin{proof}
The correctness of the algorithm follows directly from the discussion above.
\end{proof}